\documentclass[10 pt, conference]{IEEEtran}
\IEEEoverridecommandlockouts                       
\usepackage{balance}
\usepackage{url}
\usepackage[utf8]{inputenc}
\usepackage[T1]{fontenc}
\usepackage{amsmath,amssymb,amsfonts}
\usepackage{graphicx}
\usepackage{subfigure}
\usepackage{textcomp}
\usepackage{xcolor}
\usepackage{verbatim}
\usepackage{makecell}
\usepackage{booktabs}
\usepackage{cite}
\usepackage{caption}
\usepackage{extarrows}
\usepackage{bm}
\usepackage{extarrows}
\usepackage{lettrine}
\usepackage[implicit=false]{hyperref}
\usepackage{amsthm}
\usepackage{algorithm,algorithmic}
\usepackage{underscore}

\newtheorem{proposition}{Proposition}

\begin{document}
\title{\huge
Blind Interference Suppression for IRS-Aided \\Robust Wireless Communications} 

\author{
\large{Tao Wang} \href{https://orcid.org/0000-0002-8695-5400}{\includegraphics[scale=0.08]{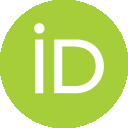}}, 
\large{Xiaohui Zhang}  \href{https://orcid.org/0009-0007-1091-4446}{\includegraphics[scale=0.08]{LIS/figures/orcid_icon.png}},
\large{Hehe Ban} \href{https://orcid.org/0009-0006-4691-730X}{\includegraphics[scale=0.08]{LIS/figures/orcid_icon.png}}, 
\large{Yiwei Guo}  \href{https://orcid.org/0009-0007-2364-7973}{\includegraphics[scale=0.08]{LIS/figures/orcid_icon.png}},
\large{and Ming Yi}

\thanks{This work was supported in part by the Mobile Information Networks National Science and Technology Major Project under Grant 2025ZD1303100, and the Henan Provincial Major Science and Technology Project under Grant 241110210300.}
\thanks{Tao Wang, Xiaohui Zhang, Hehe Ban, and Yiwei Guo are with the Songshan Laboratory, Zhengzhou, China (emails: taowang@songshanlab.com; zhangxiaohui@songshanlab.com; 1720826586@qq.com; 3830875@qq.com).} 
\thanks{Ming Yi is with PLA Strategic Support Force Information Engineering University, Zhengzhou, China (email: acco666666@sina.com.cn).}
}


\maketitle

\thispagestyle{empty}
\pagestyle{empty}

\begin{abstract}
The application of intelligent reflecting surfaces (IRSs) to suppress interference in wireless communication systems has recently attracted significant research attention. Most existing approaches rely on complete or partial channel state information (CSI) to configure the IRS. However, acquiring accurate CSI in IRS-assisted systems involves considerable pilot overhead and introduces non-negligible delays. This issue is further exacerbated under strong interference conditions, where interfering sources are typically non-cooperative, making CSI acquisition even more challenging. As a result, existing CSI-dependent interference suppression methods become difficult to deploy in practice.
To address these limitations, we propose a novel blind interference suppression strategy that combines a proportional phase-inversion (PPI) algorithm with the conditional sample mean (CSM) method. The proposed approach determines the IRS configuration using only the received signal power, without requiring any prior CSI.
We conduct a comprehensive performance evaluation by deriving the theoretical performance of the proposed scheme, which is subsequently verified through numerical simulations. Furthermore, simulation results across various parameter settings demonstrate that the proposed blind interference suppression scheme reduces the interference power to the level of noise, thereby achieving a marked signal-to-interference-plus-noise ratio (SINR) improvement and outperforms existing benchmark schemes.
\end{abstract}



\section{Introduction}
Intelligent reflecting surfaces (IRSs), also referred to as reconfigurable intelligent surfaces (RISs), are recognized as a promising technology for interference suppression in next-generation wireless communication systems \cite{Jamming_Attacks_survey_2022,IRS_Anti_Jamming_survey_2025,PLS_IRS_2024}. By dynamically adjusting the phase shifts of individual reflecting elements, an IRS can destructively combine interference signals while constructively enhancing the desired signal, thereby improving the overall signal-to-interference-plus-noise ratio (SINR) \cite{You-IRS-tutorial}. Compared to conventional interference mitigation architectures such as multiple-input multiple-output (MIMO), IRS features advantages including low hardware cost, minimal energy consumption, and ease of deployment, rendering it an efficient and sustainable solution for modern secure communication systems \cite{IRS_Anti_Jamming_survey_2025}.

Many studies have explored IRS-aided interference suppression under various system models and assumptions \cite{anti_jamming_DQN_2022,railway_anti_jamming_2021,Max_Min_Fairness_anti_jamming_2024,ITSN_anti_jamming_2024,RL_anti_jamming_CSI_free_2023}. For instance, the work in \cite{anti_jamming_DQN_2022} assumes that strong interference cannot be suppressed by the IRS once activated, and instead adopts a deep Q-network (DQN) to learn the interference pattern and optimize the transmission power and IRS phase shifts accordingly. 
The study in \cite{railway_anti_jamming_2021} introduces an RIS-assisted anti-interference method for railway communications, providing an optimal beamforming solution under perfect channel state information (CSI) and a suboptimal one when the interference channel is unknown, where only the legitimate signal power is maximized by ignoring the interference. A deep reinforcement learning (DRL) approach is developed in \cite{Max_Min_Fairness_anti_jamming_2024} to maximize the minimum SINR in a multi-user system, demonstrating robustness to imperfect interference CSI. The work in \cite{ITSN_anti_jamming_2024} proposes an optimization-driven DRL algorithm for integrated terrestrial-satellite networks, deriving a performance lower bound to account for interference uncertainty. In scenarios where the interference channel is unavailable due to its non-cooperative nature, \cite{RL_anti_jamming_CSI_free_2023} designs a DRL-based framework for UAV trajectory and RIS configuration using only received data rate feedback. 

Despite their contributions, the above studies generally assume perfect knowledge of the legitimate channel \cite{anti_jamming_DQN_2022,railway_anti_jamming_2021,Max_Min_Fairness_anti_jamming_2024,ITSN_anti_jamming_2024,RL_anti_jamming_CSI_free_2023}. Many also presume partial knowledge of the interference channel \cite{Max_Min_Fairness_anti_jamming_2024,ITSN_anti_jamming_2024}. Although CSI acquisition techniques for IRS-assisted systems have been widely studied \cite{Two_Timescale_Channel_Estimation_Dai,Sensing_RIS_Independent_CSI_Acquisition_2023,Parametric_CE_2024, Blind_CE_2024, Blind_CE_2025} along with feedback mechanisms \cite{DL_4_CSI_feedback_2025}, obtaining accurate CSI under aggressive interference in dynamic environments remains challenging for the following reasons.
\begin{enumerate}
\item The non-cooperative nature of interference makes pilot-based estimation of the interference channel infeasible. Geometry-based channel modeling \cite{railway_anti_jamming_2021,Max_Min_Fairness_anti_jamming_2024,ITSN_anti_jamming_2024,RL_anti_jamming_CSI_free_2023} offers an alternative but is limited to specific propagation environments (e.g., it is less applicable in rich-scattering scenarios or at lower carrier frequencies).
\item Estimating the legitimate channels (including the direct and IRS-reflected paths) requires considerable pilot overhead and delay even in interference-free conditions. Under strong interference, the desired signal can be overwhelmed, leading to estimation delays that exceed the channel coherence time, thereby rendering the acquired CSI obsolete.
\item The computational complexity of estimating IRS-related channels is high, often involving large-scale matrix operations. Moreover, such estimation is not yet supported by current network protocols.
\end{enumerate}

In light of these practical constraints, several CSI-free or blind beamforming methods have been proposed to configure the IRS without explicit CSI \cite{RFocus,CSM_2023,ZO_conference_2025, ZO_magazine, CS_conference_2025,CS_journal_2025}. 
For example, \cite{RFocus} introduced a method based on the conditional sample mean (CSM) of the received power to determine the ON/OFF state of each reflecting element. This approach was later extended to IRSs with discrete phase shifts in \cite{CSM_2023}, showing that the CSI-free CSM method can achieve performance close to that of the CSI-based closest point projection (CPP) method, which rounds the ideal continuous phase shifts to the nearest discrete values. 
To extend the concept of CSI-free operation to interference suppression, the work in \cite{ZO_conference_2025} introduced the zeroth-order (ZO) optimization method for blind interference suppression in IRS-aided systems, which replaces the standard gradient with a ZO gradient. Additionally, the authors of \cite{CS_conference_2025,CS_journal_2025} proposed using only the received signal power to estimate a transformed low-dimensional covariance matrix in place of the original channel covariance matrix, thereby more efficiently solving the interference neutralization problem.
However, these blind interference suppression schemes presented in \cite{ZO_conference_2025, ZO_magazine, CS_conference_2025,CS_journal_2025} all rely on the assumption of ideal continuous phase and amplitude control at the IRS. Under the practical IRS constraints of discrete phase shifts and constant modulus, the ZO scheme from \cite{ZO_conference_2025} may fail to converge \cite{ZO_magazine}, while the adaptive beamforming approach in \cite{CS_conference_2025,CS_journal_2025} suffers from significant performance degradation and often lacks a feasible solution.
To overcome these limitations, we propose a blind interference suppression scheme for IRS-aided systems under the practical constraints of discrete phase and constant modulus. The main contributions of this work are summarized as follows.

\begin{itemize}
\item We uncover the inherent potential of the IRS for interference suppression by deriving the relationship among the carrier wavelength, the IRS deployment distance, and the required number of IRS elements.
\item We introduce a novel blind interference suppression strategy that combines the proportional phase-inversion (PPI) algorithm with the CSM method. This approach configures the IRS using only the received signal power, without requiring any prior CSI or knowledge of channel distributions. Moreover, the proposed strategy accommodates practical IRS implementations with discrete and imperfect reflection coefficients, and scales efficiently to multi-user scenarios. The computational complexity of the scheme is also shown to be low.
\item We derive the theoretical performance of the proposed scheme and validate it through numerical simulations. Extensive simulations under various parameter settings further demonstrate that our scheme can suppress interference to the noise power level, thereby significantly improving the SINR and outperforming existing benchmark schemes under typical interference conditions.
\end{itemize}

The remainder of this paper is organized as follows. Section~\ref{sec:system_model} introduces the system model. Section~\ref{sec_Deployment_Settings} examines practical deployment considerations for IRS-enabled interference suppression.
Section~\ref{sec_Proposed_Scheme} details the proposed blind interference suppression strategy. Section~\ref{Sec_Performance_Analysis} provides a theoretical analysis of performance and complexity. Numerical results are discussed in Section~\ref{sec_numerical_simulation}, and Section~\ref{sec_conclusion} concludes the paper.

\textit{Notations}: 
Variables and column vectors are denoted by normal-face letters (e.g., $x$) and bold-face lower-case letters (e.g., $\mathbf{x}$), respectively. 
$\widehat{E}[x]$ represents the sampling mean of a variable. $\mathbb{E}[x]$ and $\mathrm{Var}[x]$ represent the expectation and variance of a variable, respectively.
$|x|$ and $\angle x$ extract the amplitude and phase of a variable, respectively. Lastly, $\mathcal{CN}(0, \sigma^2)$ represents the circularly symmetric complex Gaussian distribution. The most important variables and their physical meanings are listed in Table \ref{Main Symbols}.

The reproducible code for our numerical simulations is available at: \url{https://gitee.com/sssystaowang/blind_interference_suppression_2025}.

\begin{table}[t]
	\centering
	\caption{List of variables and their physical meanings}
	\label{Main Symbols}
    \renewcommand{\arraystretch}{1.5} 
	\begin{tabular}{|c|l|}
		\hline
		\textbf{Symbol} & \textbf{Physical Meaning} \\
		\hline
		$a_{S,m,n}$ & Binary selection indicator for $\mathrm{IRS}_n$ \\
		\hline
		$c$ & Constant defined as $c = \sqrt{\pi} \frac{\sin(\omega/2)}{\omega}$ \\
		\hline
		$g_0$ & Direct interference channel coefficient \\
		\hline
		$g_n$ & Equivalent cascade interference channel via $\mathrm{IRS}_n$ \\
		\hline
		$h_0$ & Direct legitimate channel coefficient \\
		\hline
		$h_n$ & Equivalent cascade legitimate channel via $\mathrm{IRS}_n$\\
		\hline
		$K$ & Number of discrete phase levels per IRS element \\
		\hline
		$M$ & Number of trials per proportion in PPI algorithm\\
		\hline
		$N$ & Number of IRS elements \\
		\hline
		$P_1$ & Average power of transmitted symbol $X$ from AP \\
		\hline
		$P_2$ & Average power of interfering symbol $Z$ \\
		\hline
		$s$ & Search step size for proportion $S$ in PPI algorithm \\
		\hline
		$S$ & Proportion of activated IRS elements \\
		\hline
		$T$ & Number of samples in CSM algorithm \\
		\hline
		$V$ & Noise \\
		\hline
		$W$ & Interference plus noise \\
		\hline
		$X$ & Transmitted symbol from AP \\
		\hline
		$Y$ & Received signal at the UE \\
		\hline
		$Z$ & Interfering symbol \\
		\hline
		$\theta_n$ & Phase shift introduced by $\mathrm{IRS}_n$ \\
		\hline
		$\boldsymbol{\theta}$ & IRS phase-shift vector $[\theta_1, \theta_2, \dots, \theta_N]$ \\
		\hline
		$\boldsymbol{\theta}^\mathrm{CSM}$ & IRS configuration obtained from CSM algorithm \\
		\hline
		$\boldsymbol{\theta}_{S,m}$ & Phase configuration for trial $m$ and proportion $S$ \\
		\hline
		$\sigma^2$ & Noise power \\
		\hline
		$\rho$ & Standard deviation of reflection channel coefficients \\
		\hline
		$\phi_n$ & Phase of reflected interference channel component \\
		\hline
		$\Phi_K$ & Discrete phase set $\{\omega, 2\omega, \ldots, K\omega\}$ \\
		\hline
		$\omega$ & Phase interval $\omega = 2\pi/K$ \\
		\hline
	\end{tabular}
\end{table}
\section{System Model}
\label{sec:system_model}
\begin{figure}[t]
\centering
\includegraphics[width=0.4\textwidth]{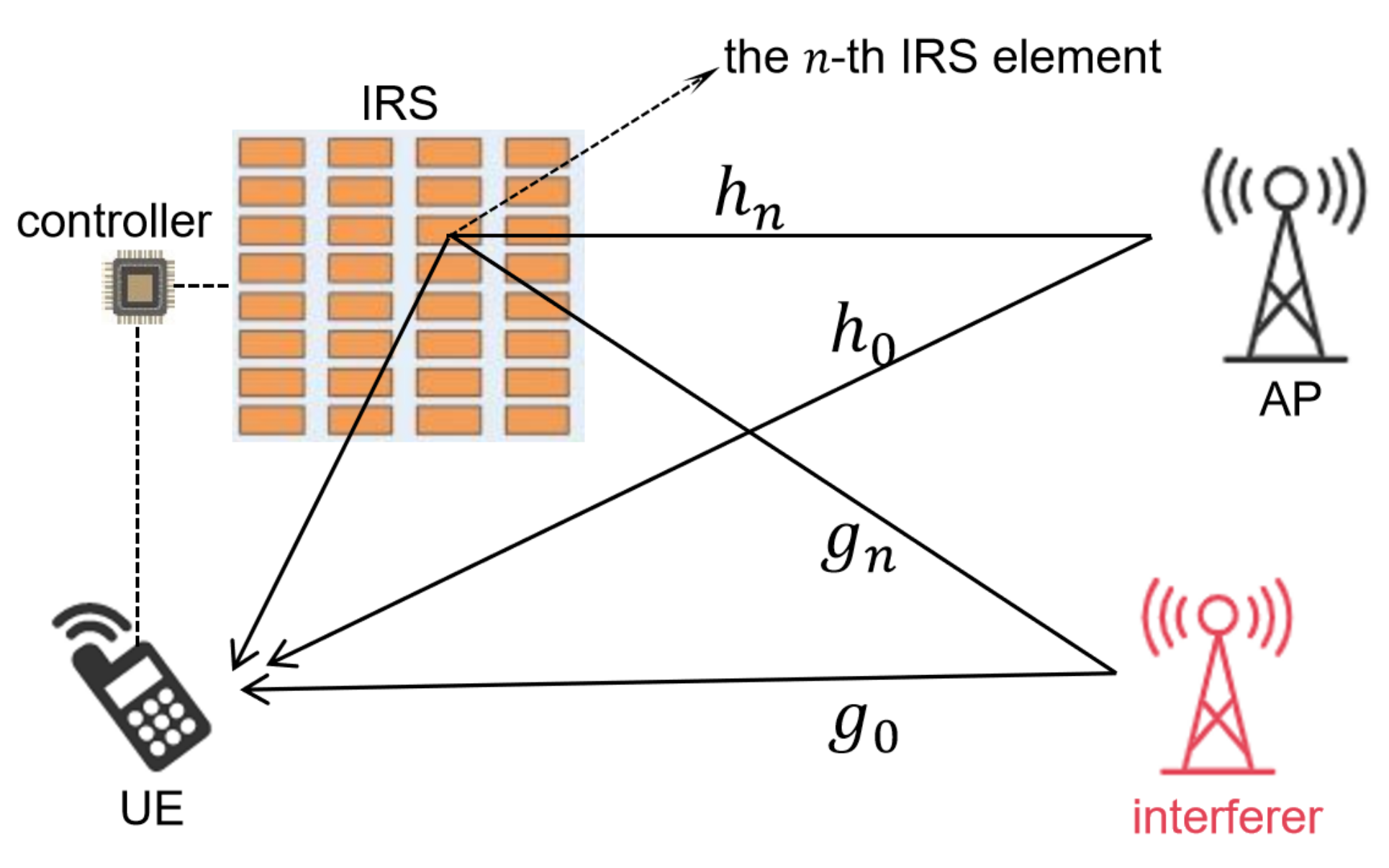}
\caption{An IRS-aided wireless communication system with interference.}
\label{fig-system-model}
\end{figure}

As illustrated in Fig.~\ref{fig-system-model}, we consider a communication system where an IRS composed of $N$ passive reflecting elements is deployed near a user equipment (UE)\footnote{The proposed scheme can be extended to multi-UE scenarios, which is elaborated in Remark \ref{remark_multi_UE} and validated in the simulations of Section \ref{sec_numerical_simulation}.} to enhance transmission quality in the presence of an interferer\footnote{The interference discussed in this paper refers to disturbances from uncontrollable sources, such as non-cooperative devices or natural background radiation, as opposed to controllable multi-user or inter-cell interference.}. The IRS is controlled by the UE\footnote{In the proposed scheme, the IRS adjusts its configuration solely based on commands from the UE, without providing feedback to either the UE or AP. This simplified architecture facilitates seamless integration with existing networks \cite{CSM_2023}.}. Let $\mathcal{N} = \{1, 2, ..., N\}$ index the IRS elements, with $\mathrm{IRS}_n$ denoting the $n$-th element. The received signal at the UE can be expressed as:

\begin{equation}
\label{eq:received_signal}
Y = \left(h_{0} + \sum_{n=1}^{N} h_{n} e^{j \theta_{n}}\right) X + \left(g_{0} + \sum_{n=1}^{N} g_{n} e^{j \theta_{n}}\right) Z + V,
\end{equation}
where $X$ is the transmitted symbol from the access point (AP) with average power $P_{1}$, i.e., $\mathbb{E}\left[|X|^{2}\right] = P_{1}$. $h_{0}$ denotes the direct channel coefficient from the AP to the UE, and $h_{n}$ denotes the equivalent cascade channel via $\mathrm{IRS}_n$ (AP-$\mathrm{IRS}_n$-UE). 
Similarly, $Z$ is the interfering symbol with average power $P_{2}$, i.e., $\mathbb{E}\left[|Z|^{2}\right] = P_{2}$. $g_{0}$ is the direct interference channel coefficient, and $g_{n}$ denotes the equivalent cascade interference channel (interferer–$\mathrm{IRS}_n$–UE). The phase shift introduced by $\mathrm{IRS}_n$ is denoted by $\theta_{n} \in [0, 2\pi)$, and $V\sim \mathcal{CN}\left(0, \sigma^{2}\right)$ represents additive white Gaussian noise (AWGN).
In accordance with typical IRS hardware constraints, each phase shift $\theta_{n}$ is selected from a discrete set:
\begin{equation}
\label{eq:phase_set}
\theta_n \in \Phi_{K} = \{\omega, 2\omega, \ldots, K\omega\}, \text{~where} \quad \omega = \frac{2\pi}{K},
\end{equation}
with $K$ being the number of discrete phase levels available per element.

When the AP is silent, the received signal contains only the interference plus noise components. Denoting this signal as $W$, we have:
\begin{equation}
\label{eq:interference_signal_system_model}
W = \left(g_{0} + \sum_{n=1}^{N} g_{n} e^{j \theta_{n}}\right) Z + V.
\end{equation}
The achievable SINR at the UE during data transmission is therefore given by:
\begin{equation}
\begin{split}
\mathrm{SINR} &= \frac{\mathbb{E}\left[\left| Y - W \right|^{2}\right]}{\mathbb{E}\left[|W|^{2}\right]} \\
    &= \frac{P_{1} \left| h_{0} + \sum_{n=1}^{N} h_{n} e^{j\theta_{n}} \right|^{2}}{P_{2} \left| g_{0} + \sum_{n=1}^{N} g_{n} e^{j\theta_{n}} \right|^{2} + \sigma^{2}}.
\end{split}
\label{eq:sinr}
\end{equation}
The objective of the proposed scheme is to determine the IRS phase-shift vector \( \boldsymbol{\theta} = [\theta_1, \theta_2, \dots, \theta_N] \) so as to minimize the interference power \( \mathbb{E}\left[|W|^{2}\right] \), without requiring any CSI.

\section{Deployment Settings for IRS-enabled Interference Suppression}
\label{sec_Deployment_Settings}
This section examines the inherent potential of the IRS for interference suppression by analyzing the relationship between the IRS deployment distance and the number of IRS elements required for effective interference suppression.

From (\ref{eq:sinr}), to eliminate the interference power, the condition $\sum_{n=1}^{N} g_{n} e^{j\theta_{n}} = -g_{0}$ must be satisfied. Under the ideal assumption of continuous IRS phase shifts, each reflection coefficient $\theta_n$ can be configured such that every reflecting channel aligns with $-g_0$, i.e., $\theta_n = \angle(-g_0)-\angle(g_n)$. In this case, achieving perfect interference suppression requires the sum of the reflection channel amplitudes to equal the direct channel amplitude: $\sum_{n=1}^{N} |g_{n}| = |g_{0}|$. Considering practical constraints of discrete phase control and generally suboptimal optimization, a more robust condition for effective interference suppression is
\begin{equation}
\label{eq:gn>go}
\sum_{n=1}^{N} |g_{n}| \ge |g_{0}|.
\end{equation} 
Based on (\ref{eq:gn>go}), the following analysis derives the relationship between the deployment distance and the requisite number of IRS elements under practical channel assumptions.

The channel for the interferer–$\mathrm{IRS}_n$–UE link can be modeled as
\begin{equation}
\label{eq:gn}
g_n = g_{n,1} \cdot \sqrt{A_{n, \rm{IRS}}}  \cdot g_{n,2} \cdot \sqrt{A_{n, \rm{UE}}},
\end{equation}
where where $g_{n,1}$ denotes the path loss from the interferer to $\mathrm{IRS}_n$, and $g_{n,2}$ denotes the path loss from $\mathrm{IRS}_n$ to the UE. The terms $A_{n, \rm{IRS}}$ and $A_{n, \rm{UE}}$ represent the effective apertures of $\mathrm{IRS}_n$ and the UE, respectively, corresponding to the cascaded channel $g_n$. Similarly, the direct interferer–UE channel is
\begin{equation}
\label{eq:g0}
g_0 = g_{0,1} \cdot \sqrt{A_{0, \rm{UE}}},
\end{equation}
where $g_{n,1}$ is the path loss from the interferer to the UE, and $A_{0, \rm{UE}}$ is the UE's effective aperture for the direct path $g_0$.

We consider a typical deployment where the IRS is placed near the UE, while the interferer is located at a significantly greater distance. Accordingly, the following assumptions are made:
\begin{enumerate}
    \item The path losses $g_{n,1}, n\in \mathcal{N}$ and $g_{0,1}$ follow an identical distribution.
    \item The channel$g_{n,2}$ is dominated by its line-of-sight (LoS) component, approximated as: \[|g_{n,2}| \approx \frac{1}{\sqrt{4\pi d^2}},\] where $d$ is the $\mathrm{IRS}$–UE distance.
    \item Without loss of generality, we assume $A_{n, \rm{UE}} = A_{0, \rm{UE}}, n\in \mathcal{N}$. Given that typical IRS elements are of sub-wavelength size, we set $A_{n, \rm{IRS}} = \lambda^2/4$.
\end{enumerate}

Applying (\ref{eq:gn}), (\ref{eq:g0}) and the above assumptions yields 
\begin{equation}
\label{eq:g0_gn}
\begin{split}
\frac{\widehat{E}[|g_0|]}{\widehat{E}[|g_n|]} &= \frac{\widehat{E}[|g_{0,1}|] \cdot \sqrt{A_{0, \rm{UE}}}}{\widehat{E}[|g_{n,1}|] \cdot \sqrt{A_{n, \rm{IRS}}}  \cdot \widehat{E}[|g_{n,2}|] \cdot \sqrt{A_{n, \rm{UE}}}}\\
&= \frac{1}{\sqrt{A_{n, \rm{IRS}}}  \cdot \widehat{E}[|g_{n,2}|]}\\
&=\frac{4\sqrt{\pi} d}{\lambda}.   
\end{split}
\end{equation}
This indicates that, on average, the magnitude $|g_0|$ is $\frac{4\sqrt{\pi} d}{\lambda}$ times larger than $|g_n|$. Combining this result with condition (\ref{eq:gn>go}), the minimum number of IRS elements required is
\begin{equation}
\label{eq:N_min}
\begin{split}
N_{\rm{min}} = \frac{4\sqrt{\pi} d}{\lambda},
\end{split}
\end{equation}
which is directly proportional to the IRS-UE distance $d$ and inversely proportional to the carrier wavelength $\lambda$. 

\begin{figure}[t]
\centering
\includegraphics[width=0.5\textwidth]{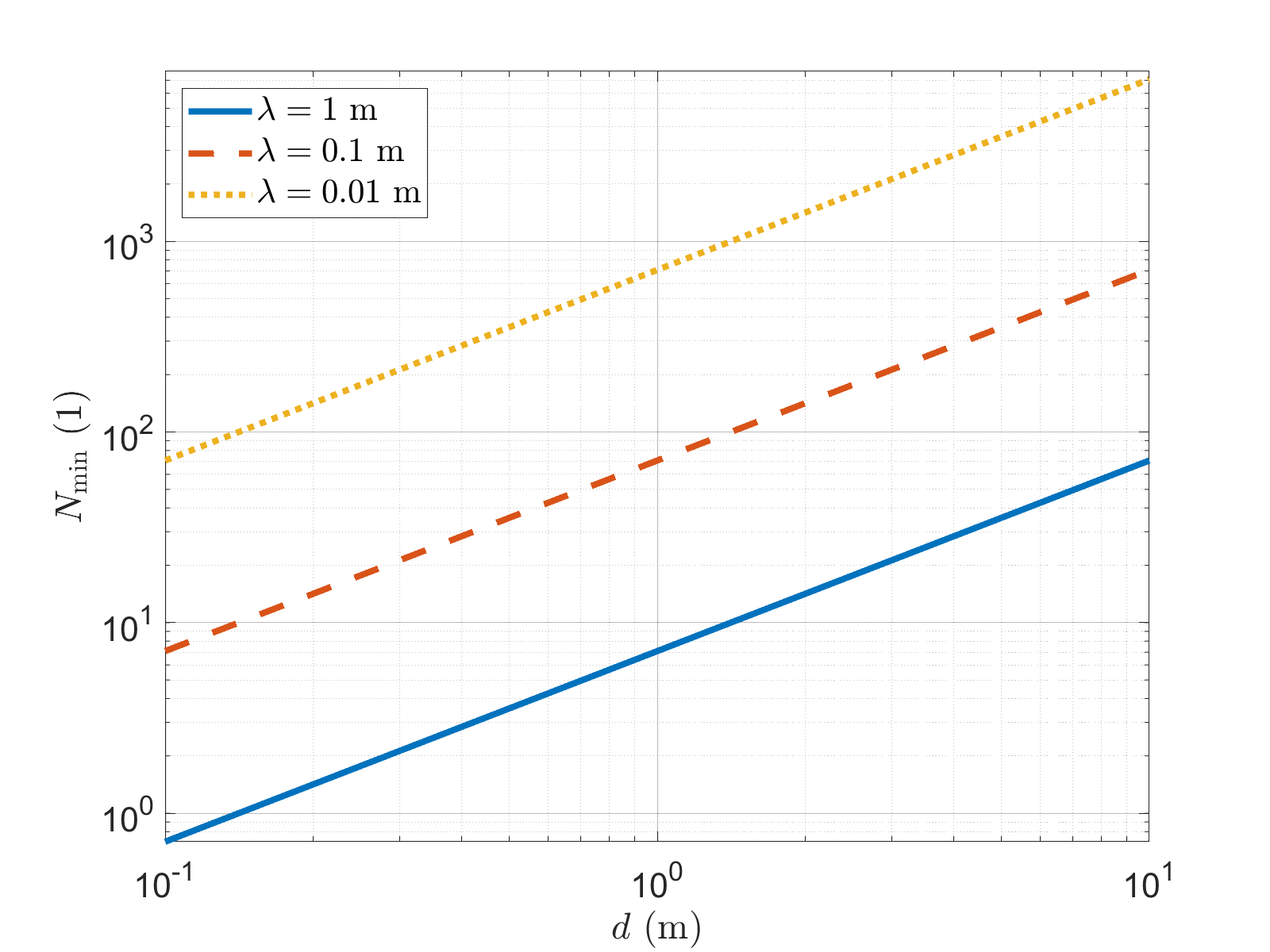}
\caption{Minimum required number of IRS elements vs. IRS-UE distance, under various carrier frequencies.}
\label{fig-Amplitude_Ratio}
\end{figure}

Fig. \ref{fig-Amplitude_Ratio} plots plots $N_{\rm{min}}$ versus $d$ for different values of the carrier wavelength $\lambda$. It is shown that for $d = 0.1$ m, $N_{\rm{min}}$ is approximately 0.7, 7, and 70 for the 0.3 GHz, 3 GHz, and 30 GHz bands, respectively. While these numbers are small, such a near-field deployment may induce model mismatch with conventional channel modeling. 
By contrast, for $d = 10$ m, $N_{\rm{min}}$ is approximately 70, 700, and 7000 for the respective frequency bands, which imposes prohibitive hardware and operational costs. A more practical distance is $d = 1$ m, where $N_{\rm{min}}$ is approximately 7, 70, and 700. This distance avoids significant near-field effects while maintaining a feasible hardware scale. In summary, the relationship in (\ref{eq:g0_gn}) provides essential guidance for the deployment of IRS-based interference suppression systems.

The preceding analysis assumes perfect CSI knowledge and continuous IRS phase shifts. The following section proposes a practical blind interference suppression scheme that operates without CSI and under discrete phase-shift constraints.

\section{Proposed Blind Interference Suppression Scheme}
\label{sec_Proposed_Scheme}
This section describes the overall procedure and key algorithms of the proposed blind interference suppression scheme. The core idea is to determine the IRS phase configuration such that the reflected interference channels combine destructively with the direct interference channel at the receiver, thereby suppressing the overall interference power.

\subsection{Overall Procedure}
\begin{figure}[ht]
\centering
\includegraphics[width=0.35\textwidth]{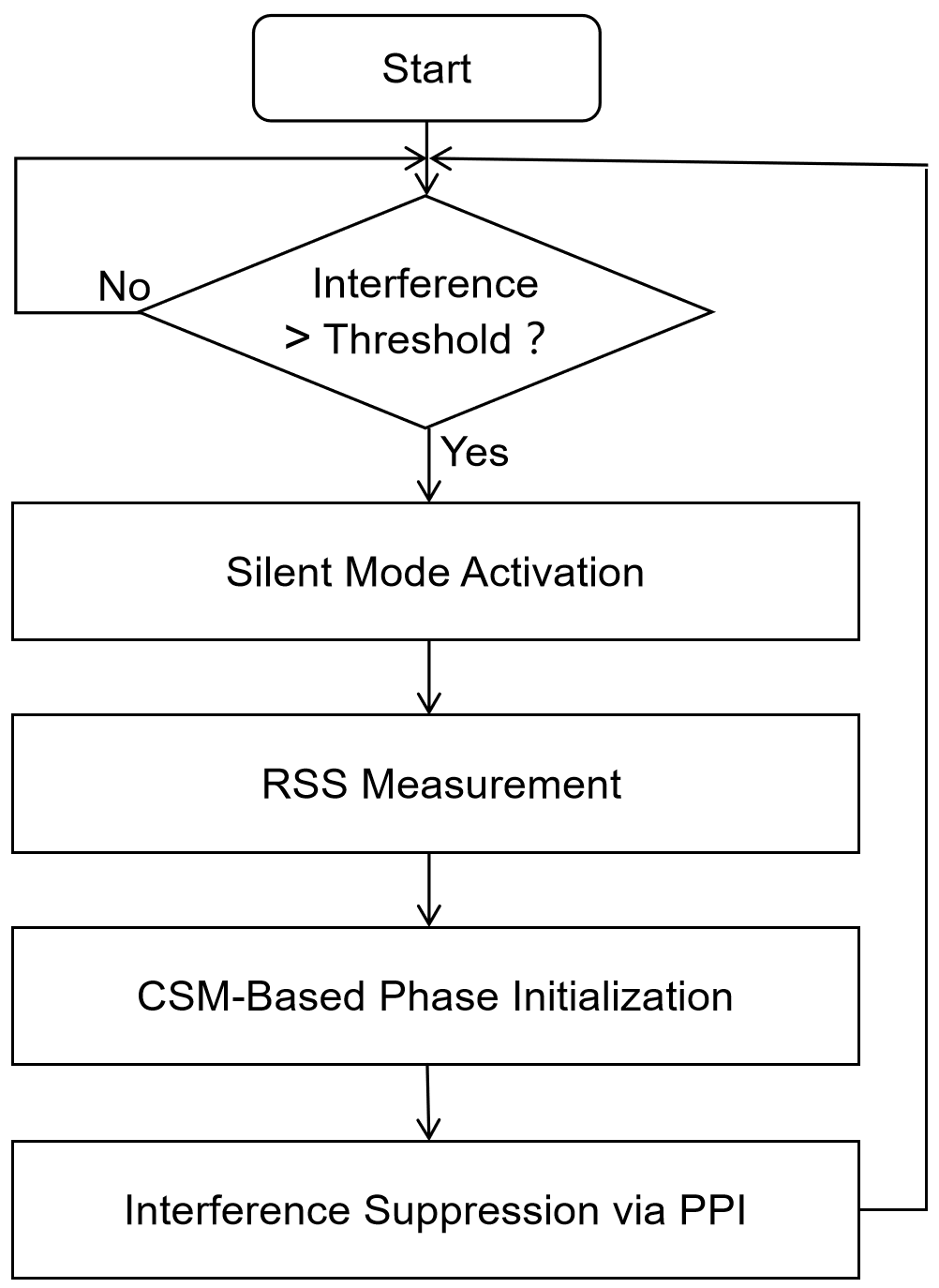}
\caption{Flowchart of the proposed blind interference suppression scheme.}
\label{fig:flowchart}
\end{figure}

The overall flowchart of the proposed scheme is illustrated in Fig.~\ref{fig:flowchart}, which outlines a cyclic process of interference detection, measurement, and suppression. The scheme operates through the following four main steps:
\begin{enumerate}
    \item \textbf{Silent Mode Activation}: During normal communication, the AP temporarily enters a silent mode upon request from the UE when the measured interference power exceeds a predefined threshold, which is also assumed in prior works \cite{ZO_conference_2025, CS_conference_2025}. Alternatively, silent periods can be introduced during zero-power symbol intervals in pilot sequences\cite{3GPP.38.214,3GPP.38.331}\footnote{To counter proactive interferers that transmit only when sensing legitimate signals, the AP can insert zero-power symbols in its pilot sequence. This ensures that the interferer remains active during the AP's silent periods, allowing interference measurement.}.
    \item \textbf{RSS Measurement}: While the AP is silent, the UE configures the IRS with random phase-shift vectors over $T$ independent trials. Denote the phase configuration in the $t$-th trial by $\boldsymbol{\theta}_t = \{\theta_{1t}, \theta_{2t}, \dots, \theta_{Nt}\}$, where $t = 1, 2, \dots, T$. The corresponding received signal strength (RSS), equivalent to the interference-plus-noise power $|W_t|^2$, is measured at the UE.
    \item \textbf{CSM-Based Phase Initialization}: Using the collected dataset $\{[\boldsymbol{\theta}_t, |W_t|^2]\},  t = 1, 2, \dots, T$, the CSM algorithm is applied to compute an initial phase configuration $\boldsymbol{\theta}^\mathrm{CSM}$. This configuration minimizes the phase misalignment between each IRS-reflected interference channel and the direct interference channel.
    \item \textbf{Interference Suppression via PPI}: Based on $\boldsymbol{\theta}^\mathrm{CSM}$, the PPI algorithm is employed to determine the final IRS configuration that achieves effective interference suppression.
\end{enumerate}

The process is repeated periodically to adapt to dynamic interference conditions. Unlike existing approaches, the proposed scheme operates without any CSI and does not assume any specific channel model (e.g., Rician, Rayleigh, or sparse multipath), thereby improving robustness in practical scenarios. The two core components—the CSM algorithm and the PPI algorithm—are detailed in the following subsections.

\subsection{CSM Algorithm}
The CSM algorithm is employed to obtain an initial IRS phase configuration $\boldsymbol{\theta}^\mathrm{CSM}$ that aligns the IRS-reflected interference channels with the direct interference channel. This method utilizes conditional sample means of the measured interference power for each discrete phase shift of every IRS element \cite{CSM_2023}. The detailed procedure is as follows:

\begin{enumerate}
\item \textbf{Calculate Conditional Sample Mean}: Using the collected dataset $\{[\boldsymbol{\theta}_t, |W_t|^2]\},  t = 1, 2, \dots, T$, compute for each IRS element $\mathrm{IRS}_n$ and each discrete phase $\varphi \in \Phi_K$ the conditional sample mean of $|W|^2$ given that $\theta_n = \varphi$:
\begin{equation}
   \widehat{E}[|W|^2 \mid \theta_n = \varphi] = \frac{1}{|\mathcal{T}_{n,\varphi}|} \sum_{t \in \mathcal{T}_{n,\varphi}} |W_t|^2,    
\end{equation}
where \( \mathcal{T}_{n,\varphi} = \{ t : \theta_{nt} = \varphi \} \) denotes the set of trial indices in which the phase of $\mathrm{IRS}_n$ equals \( \varphi \).

\item \textbf{Obtain $\boldsymbol{\theta}^\mathrm{CSM}$}: For each element $\mathrm{IRS}_n$, select the phase value that maximizes the corresponding conditional sample mean:
\begin{equation}
   \theta_n^{\text{CSM}} = \arg \max_{\varphi \in \Phi_K} \widehat{E}[|W|^2 \mid \theta_n = \varphi].    
\end{equation}
The resulting IRS configuration is given by:   
\begin{equation}
\label{xita-CSM}
    \boldsymbol{\theta}^\mathrm{CSM} = [\theta_1^{\text{CSM}}, \theta_2^{\text{CSM}}, \dots, \theta_N^{\text{CSM}}].
\end{equation}
\end{enumerate}

As established in \cite{CSM_2023}, the configuration \( \boldsymbol{\theta}^\mathrm{CSM} \) obtained from the CSM algorithm minimizes the phase difference between each reflected interference channel component \( g_{n} e^{j\theta_n^{\mathrm{CSM}}}, n\in\mathcal{N} \) and the direct interference channel $g_{0}$. Consequently, the phases of \( g_{n} e^{j\theta_n^{\mathrm{CSM}}}, n\in\mathcal{N} \) are distributed within an interval of \( \pm \frac{\omega}{2} \) around the phase of \( g_{0} \), as illustrated in Fig. \ref{fig:phase_CSM}.

\begin{figure}[h]
\centering
\includegraphics[width=0.4\textwidth]{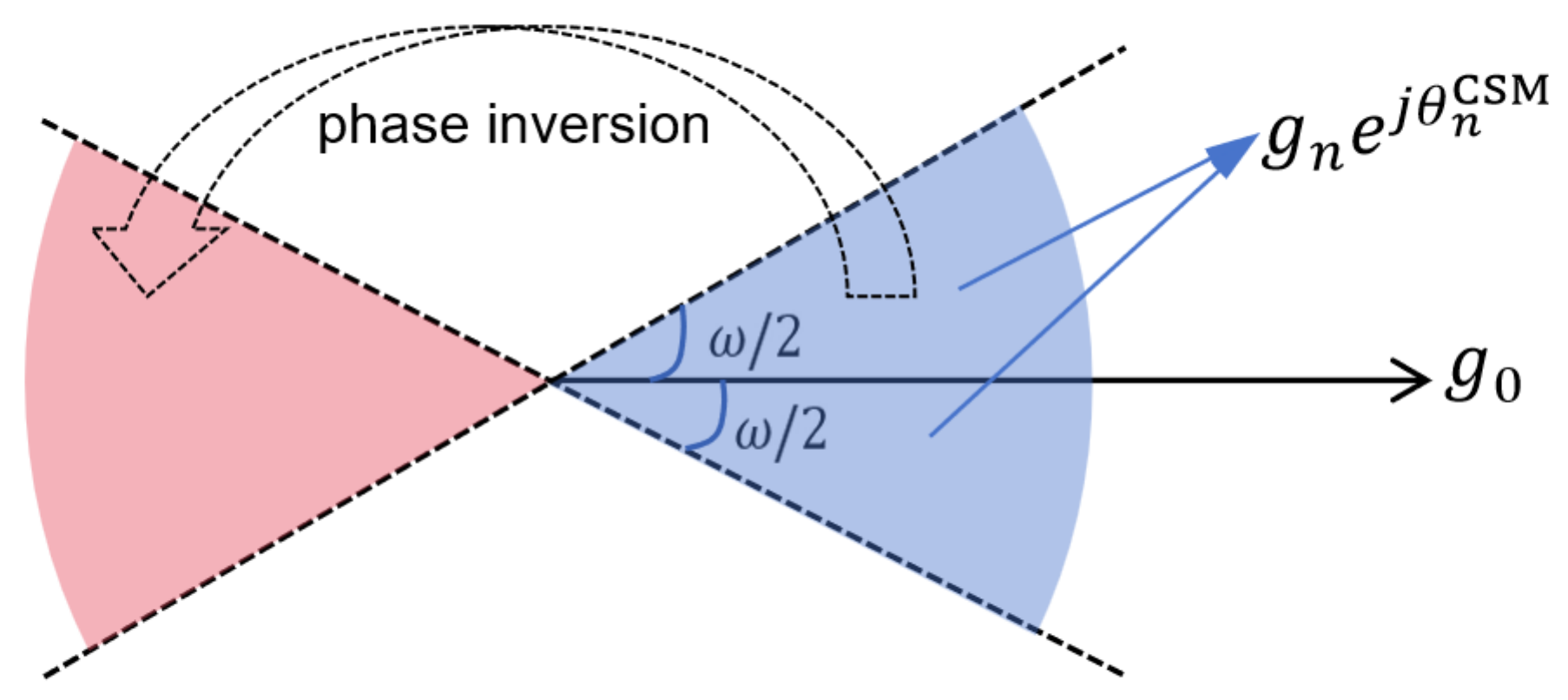}
\caption{Under the phase configuration of \( \boldsymbol{\theta}^{\text{CSM}} \), the phases of the reflected interference channel components are confined within \( \pm \frac{\omega}{2} \) of the direct channel phase, enabling subsequent interference suppression via phase inversion.}
\label{fig:phase_CSM}
\end{figure}

\subsection{PPI Algorithm}
\label{Phase_Inversion_algorithm}
Based on the initial configuration \( \boldsymbol{\theta}^{\text{CSM}} \) obtained through the CSM algorithm, the PPI algorithm is employed to achieve interference suppression. The core concept of PPI is to select a subset of IRS elements, invert the phases of these elements by adding $\pi$, and deactivate the remaining elements, thus creating a destructive superposition of the reflected channels and the direct channel. This process is vividly shown in Fig. \ref{fig:phase_CSM}. To this end, the PPI algorithm aims to search for the proportion $S\in[0,1]$ of activated elements and their specific combination that minimizes interference power. The detailed procedure is described below.
\begin{enumerate}
\item \textbf{Define Trial Parameters}: Define the set of candidate proportions as $\{s, 2s, \dots, 1\}$, where $s$ is the search step size. For each candidate proportion, conduct $M$ independent trials. Therefore, the total number of trials is \(\frac{M}{s} \), which remains manageable (e.g., 1000 trials for $s$=0.01 and $M$=10).
\item \textbf{Perform Trials}: For \( S \in \{s, 2s, \dots, 1\}\):
    \begin{itemize}
         \item For each trial \( m \in \{1, 2, \dots, M\} \):
        \begin{itemize}
             \item Randomly select \( \mathrm{round}(S N) \) IRS elements form the total $N$ elements. Define the binary selection indicators $\mathbf{a}_{S,m} = [a_{S,m,1}, a_{S,m,2}, ..., a_{S,m, N}]$, where
                \begin{equation}
                a_{S,m,n}=\left\{\begin{array}{cc}
                1, & \text{if~} \mathrm{IRS}_{n} \text{ is selected} \\
                0, & \text{else}
                \end{array}\right..  
                \end{equation}
             \item Activate only the selected IRS elements while turning off the others\footnote{Deactivating a portion of IRS elements reduces power consumption while maintaining interference suppression performance.}. Invert the phases of the selected elements, The resulting phase configuration for proportion $S$ and trial $m$ is given by $\boldsymbol{\theta}_{S,m} = [\theta_{S,m,1}, \theta_{S,m,2}, \dots, \theta_{S,m,N}]$, where
                \begin{equation}
                \theta_{S,m,n}=\left\{\begin{array}{cc}
                \theta_{n}^{\mathrm{CSM}}+\pi, & \text{if } a_{S,m,n}=1 \\
                \text{Deactivated}, & \text{else}
                \end{array}\right.. 
                \end{equation}
            \item The received signal for proportion $S$ and trial $m$ is
                \begin{equation}
                \label{eq:interference_signal}
                W_{S,m} = \left(g_{0} + \sum_{n=1}^{N} a_{S,m,n}g_{n} e^{j \theta_{S,m,n}}\right) Z + V.
                \end{equation}
            \item Measure the corresponding interference-plus-noise power \( |W_{S,m}|^2 \).
        \end{itemize}
    \end{itemize}
    \item \textbf{Obtain Final IRS Configuration}: After completing all trials, determine the optimal IRS configuration that minimizes the measured interference power:
        \begin{equation}
            \begin{split}
        \mathbf{a}^{\text{final}},\boldsymbol{\theta}^{\text{final}} = \arg \min_{\mathbf{a}_{S,m},~\boldsymbol{\theta}_{S,m}} |W_{S,m}|^2.        
            \end{split}
        \end{equation}  
\end{enumerate}

In summary, the proposed scheme combines CSM-based phase alignment with PPI to achieve effective interference suppression, relying solely on the received signal power without requiring any CSI.

\newtheorem{remark}{Remark}
\begin{remark}[Extension to multi-UE Scenarios]
\label{remark_multi_UE}
\rm{
The proposed scheme can be extended to multi-UE scenarios where each UE is assisted by a dedicated IRS, with each UE independently executing the blind interference suppression scheme. For a given UE, interference signals reflected by other IRSs may introduce deviations in the measured received power. Nevertheless, since the distance between an IRS and its served UE is typically much shorter than its distance to other UEs, the interference contributed by non-serving IRSs remains relatively weak and does not substantially degrade the achievable performance. The effectiveness of the proposed scheme in multi-UE scenarios is validated through simulation results presented in Fig. \ref{fig_SINR_interUEdist} in Section \ref{sec_numerical_simulation}.
}
\end{remark}

\begin{remark}[Robustness Against Imperfect IRS Reflection Coefficients]
\label{remark_error}
\rm{Practical IRS implementations exhibit non-ideal reflection coefficients due to hardware limitations such as insertion loss, quantization errors, mutual coupling, and thermal drift. The proposed scheme demonstrates inherent robustness against such imperfections. Specifically, the CSM algorithm determines the optimal IRS configuration based solely on the conditional sampling mean of the received power, without presupposing any particular model for the IRS reflection coefficients. Similarly, the PPI algorithm performs an exhaustive search over inversion ratios and randomly samples IRS element combinations for a given ratio, without requiring prior knowledge of the IRS coefficient characteristics. Consequently, the proposed approach remains effective despite unpredictable fluctuations in the reflection coefficients induced by hardware non-idealities. This robustness is validated through simulation results in Fig. \ref{fig_SINR_E} in Section \ref{sec_numerical_simulation}.
}
\end{remark}

\section{Performance Analysis}
\label{Sec_Performance_Analysis}
This section presents a theoretical analysis of the proposed scheme's performance, focusing on the Interference Power Change Ratio (IPCR) and the Legitimate Signal Power Change Ratio (LSPCR). Specifically, IPCR is defined as the ratio of the optimized interference signal power to the original interference signal power, which quantifies the relative change in the interference power resulting from optimization. Similarly, LSPCR is defined as the ratio of the optimized legitimate signal power to the original legitimate signal power. In addition, the complexity of the proposed scheme is also analyzed.

Without loss of generality, we take the phase of $g_{0}$ as the reference (i.e.,  $\angle g_{0} = 0$). The reflection channel coefficients $g_{n}$ are assumed to be independent and identically distributed as $g_{n}\sim\mathcal{CN}(0,\rho^2)$\footnote{Other distributions can be similarly analyzed with minor modifications.}. Considering the marginal channel gain contributed by a single IRS element, we assume that $g_0 \gg \rho$. After conducting the CSM and PPI algorithm, the phases of the selected reflection channels, denoted by $\phi_n = \mathrm{\angle}(g_{n} e^{j\theta_n^{\rm{CSM}+\pi}})$ for each IRS element, are assumed to be uniformly distributed over $[\pi- \omega/2, \pi + \omega/2)$, as illustrated in Fig.~\ref{fig:phase_CSM}. For analytical simplicity, we normalize $P_1 = P_2 = 1$, effectively omitting $X$ and $Z$ from the expressions. The selection indicator $a_{S,m,n}$ is simplified to $a_{n}$ in the analysis.

\subsection{Interference Power Change Ratio (IPCR)}
\label{Expected_Power_of_W}

\begin{proposition}
Given proportion $S\in[0,1]$ and a random combination of activated elements, the expected interference-plus-noise power is given by:
\begin{equation}
\begin{split}
\mathbb{E}[|W|^2] = A S^2 + B S + C.
\end{split}
\label{power_of_W_final}
\end{equation}
where $A =  \rho^2c^2 N^2,~B = N [(1 - c^2) \rho^2 - 2 g_{0} c \rho],~C = g_{0}^2 + \sigma^2$, $c = \sqrt{\pi} \frac{\sin(\omega/2)}{\omega}$. 
\end{proposition}
\begin{proof}
    See Appendix.
\end{proof}

We minimize the quadratic function $E(S) = A S^2 + B S + C$ over \(S \in [0, 1]\). The optimum occurs at:
\begin{equation}
    \begin{split}
S^* = -\frac{B}{2A} = \frac{2 g_{0} c + (c^2 - 1)\rho}{2 \rho c^2 N}.      
    \end{split}
\end{equation}
The expression of IPCR is characterized by three cases:

\begin{itemize}
\item \textbf{Case 1: $S^* < 0$}. The minimum is achieved at $S=0$, yielding: 
\begin{equation}
    \begin{split}
    \mathbb{E}[|W|^2]_{\text{min}} = g_{0}^2 + \sigma^2.    
    \end{split}
\end{equation}
This case occurs when $S^* < 0$, i.e., $\frac{2 g_{0} c + (c^2 - 1)\rho}{2 \rho c^2 N} < 0$, thus we have
\begin{equation}
    \begin{split}    
g_{0} < \frac{1-c^2}{2c}\rho.
    \end{split}
\end{equation}
Since $c = \sqrt{\pi} \frac{\sin(\omega/2)}{\omega}, \omega\in(0, \pi]$ monotonically decreases with $\omega$, $\frac{1-c^2}{2c}$ monotonically decreases with $c$, we have $\frac{1-c^2}{2c}$ monotonically increases with $\omega$. Through calculation, we have $0.12<\frac{1-c^2}{2c}\le 0.60$ for $\omega\in(0, \pi]$.
This indicates that $g_{0}$ is smaller than the expected amplitude of a single reflection channel, i.e., $\mathbb{E}[|g_{n}|] = \frac{\sqrt{\pi}}{2} \rho \approx 0.89 \rho$ (see Appendix). Therefore, $g_{0}$ can be regarded as negligible in this case, where no IRS activation is needed. Thus, $\mathrm{IPCR} = 1$ in this case. 

\item \textbf{Case 2: $S^* > 1$}. The minimum occurs at $S = 1$, giving:
\begin{equation}
    \begin{split}
       \mathbb{E}[|W|^2]_{\text{min}} &= A+B+C\\
       &=g_{0}^2 - 2 g_{0} N c\rho + N^2 c^2\rho^2 \\&~~~ -N c^2\rho^2  + N\rho^2 + \sigma^2 \\&= (g_{0}-Nc\rho)^2+N\rho^2(1-c^2)+\sigma^2.
       \label{power_not_engough_IRS_elements}
    \end{split}
\end{equation}
With equations (\ref{power_not_engough_IRS_elements}), the IPCR is:
\begin{equation}
    \begin{split}
\mathrm{IPCR} &=\frac{(g_{0}-Nc\rho)^2+N\rho^2(1-c^2)}{g_{0}^2}.
\label{IPCR_boost_insufficient}
\end{split}
\end{equation}
This case occurs when $S^* >1$, i.e., $\frac{2 g_{0} c + (c^2 - 1)\rho}{2 \rho c^2 N} > 1$, thus we have 
\begin{equation}
    \begin{split}
N< \frac{2 g_{0} c + (c^2 - 1)\rho}{2 \rho c^2}  = \frac{g_0}{c\rho}+\frac{c^2-1}{2c^2}. 
  \label{not_enough_IRS_elements}
    \end{split}
\end{equation}
This indicates that when the deployed IRS elements are insufficient, as shown in (\ref{not_enough_IRS_elements}), even with all IRS elements activated, i.e., $S = 1$, $g_{0}$ cannot be fully canceled. In this case, the number of activated IRS elements is $N$. 

\item \textbf{Case 3: $0 \leq S^* \leq 1$}. The minimum occurs at $S = S^*$, yielding:
\begin{equation}
    \begin{split}
\mathbb{E}[|W|^2]_{\text{min}} &= C-\frac{B^2}{4A} 
\\&= g_{0}^2 + \sigma^2 - \frac{\left((1 - c^2)\rho - 2 g_{0} c\right)^2}{4c^2} 
\\&= \frac{1 - c^{2}}{c}\rho g_{0} - \left(\frac{1 - c^{2}}{2c}\right)^2\rho^2 + \sigma^{2}. 
    \end{split}
\label{expected_interference_plus_noise_power}
\end{equation}
With equations (\ref{expected_interference_plus_noise_power}), the IPCR is:
\begin{equation}
    \begin{split}
\mathrm{IPCR}&=\frac{\frac{1 - c^{2}}{c}\rho g_{0} - \left(\frac{1 - c^{2}}{2c}\right)^2\rho^2}{g_{0}^2}\\
&= \frac{1 - c^{2}}{c}\frac{\rho}{g_{0}}\left(1- \frac{1 - c^{2}}{4c}\frac{\rho}{g_0}\right).
\label{SINR_boost}
\end{split}
\end{equation}
Since $S = \frac{2 g_{0} c + (c^2 - 1)\rho}{2 \rho c^2 N}$, we have 
\begin{equation}
    \begin{split}
  N = \frac{1}{S}(\frac{g_0}{c\rho}+\frac{c^2-1}{2c^2}) \ge \frac{g_0}{c\rho}+\frac{c^2-1}{2c^2}. 
  \label{amplitude_relationship}
    \end{split}
\end{equation}
This condition indicates that deploying enough IRS elements to satisfy (\ref{amplitude_relationship}) is essential for effective interference suppression. In this case, the number of activated IRS elements is $SN = \frac{g_0}{c\rho}+\frac{c^2-1}{2c^2}$. Combining equation (\ref{SINR_boost}), it can be seen that both IPCR and the number of activated IRS elements do not vary with $N$. 
\end{itemize}
\begin{figure*}[t]
    \centering
    \subfigure[$K = 2,~T = 200$.]
    {\includegraphics[scale = 0.131]{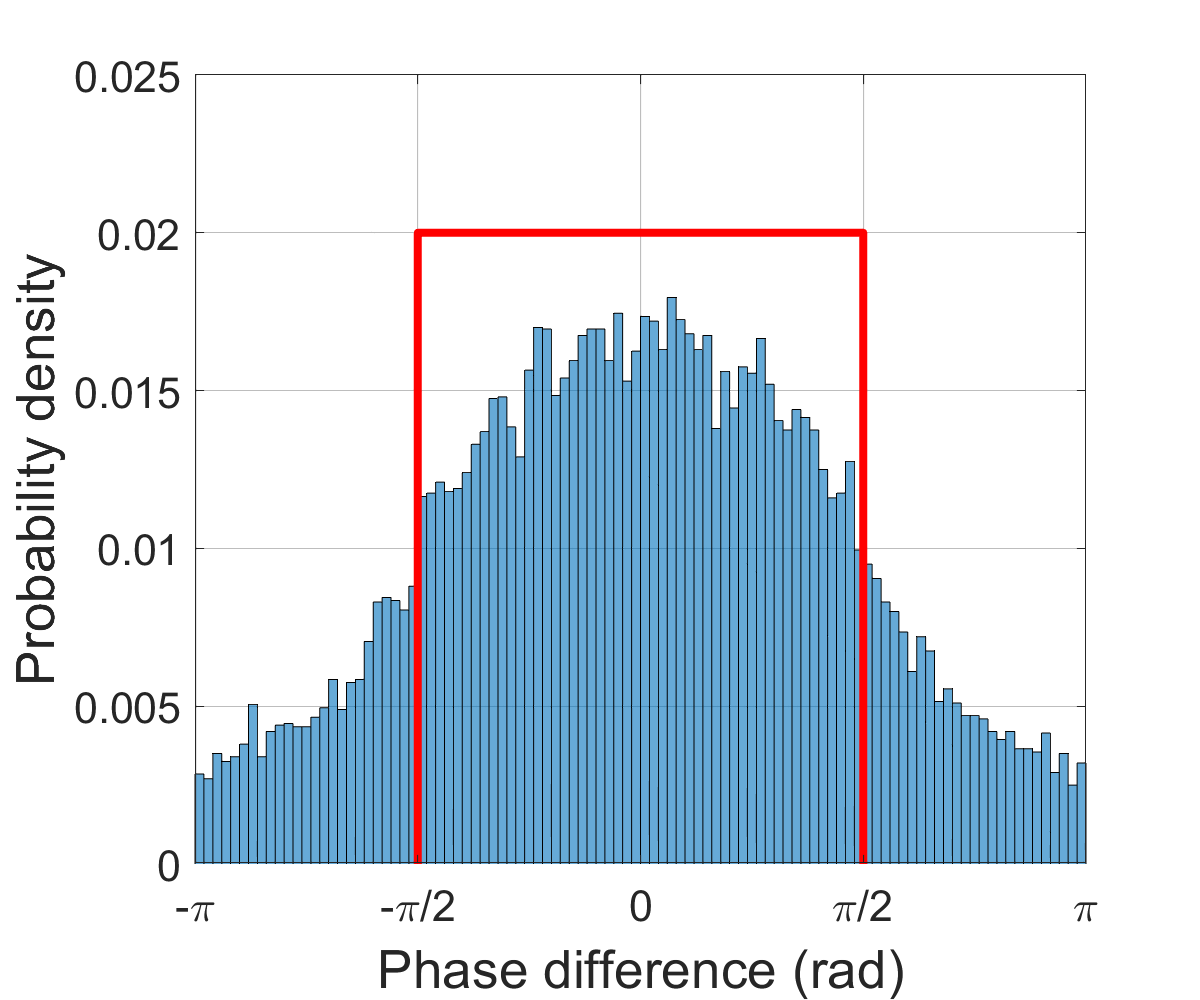}
    }     
     \subfigure[$K = 2,~T = 2000$.]
    {\includegraphics[scale = 0.131]{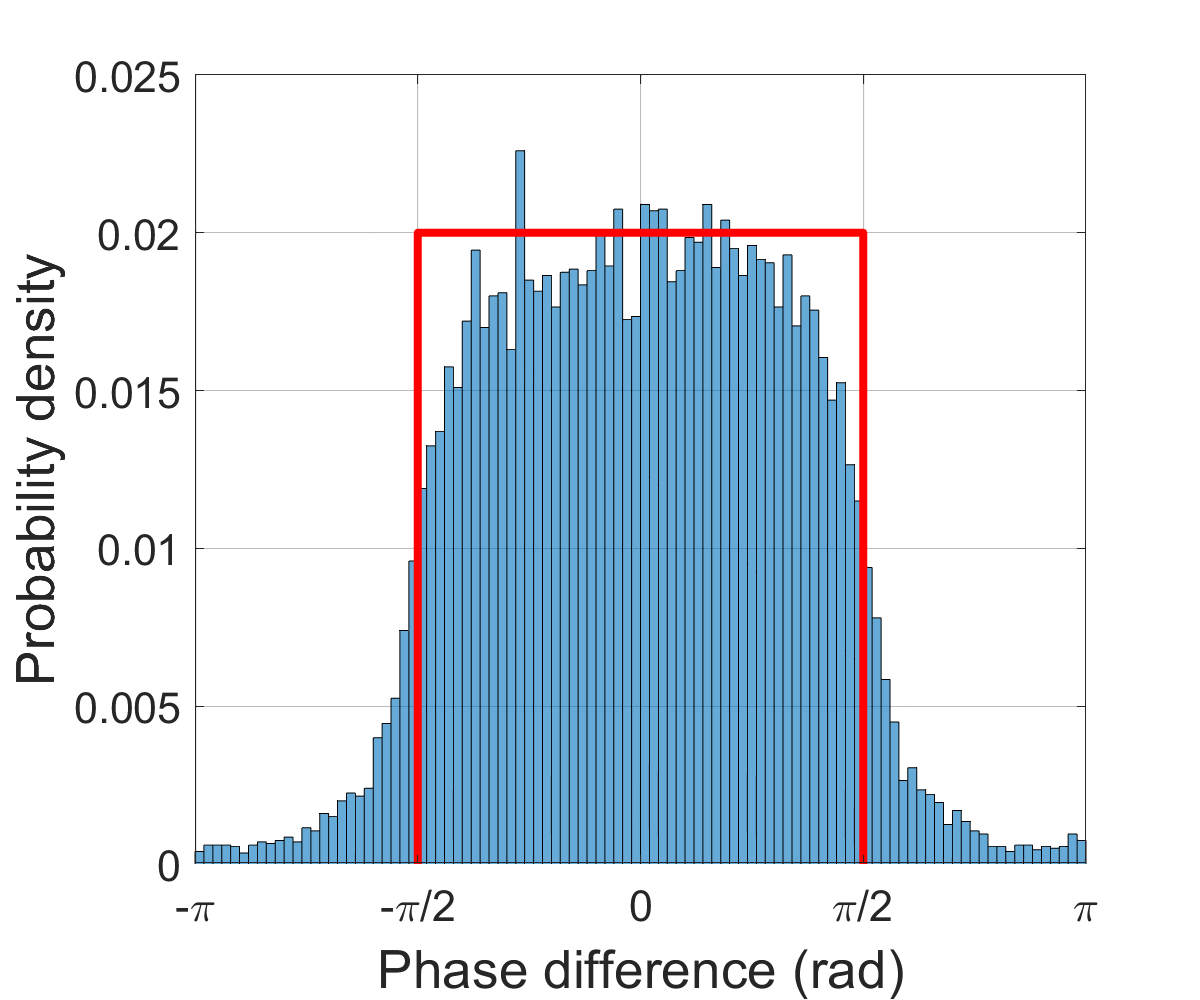}	
    }
\quad
    \subfigure[$K = 2,~T = 20000$.]
    {\includegraphics[scale = 0.131]{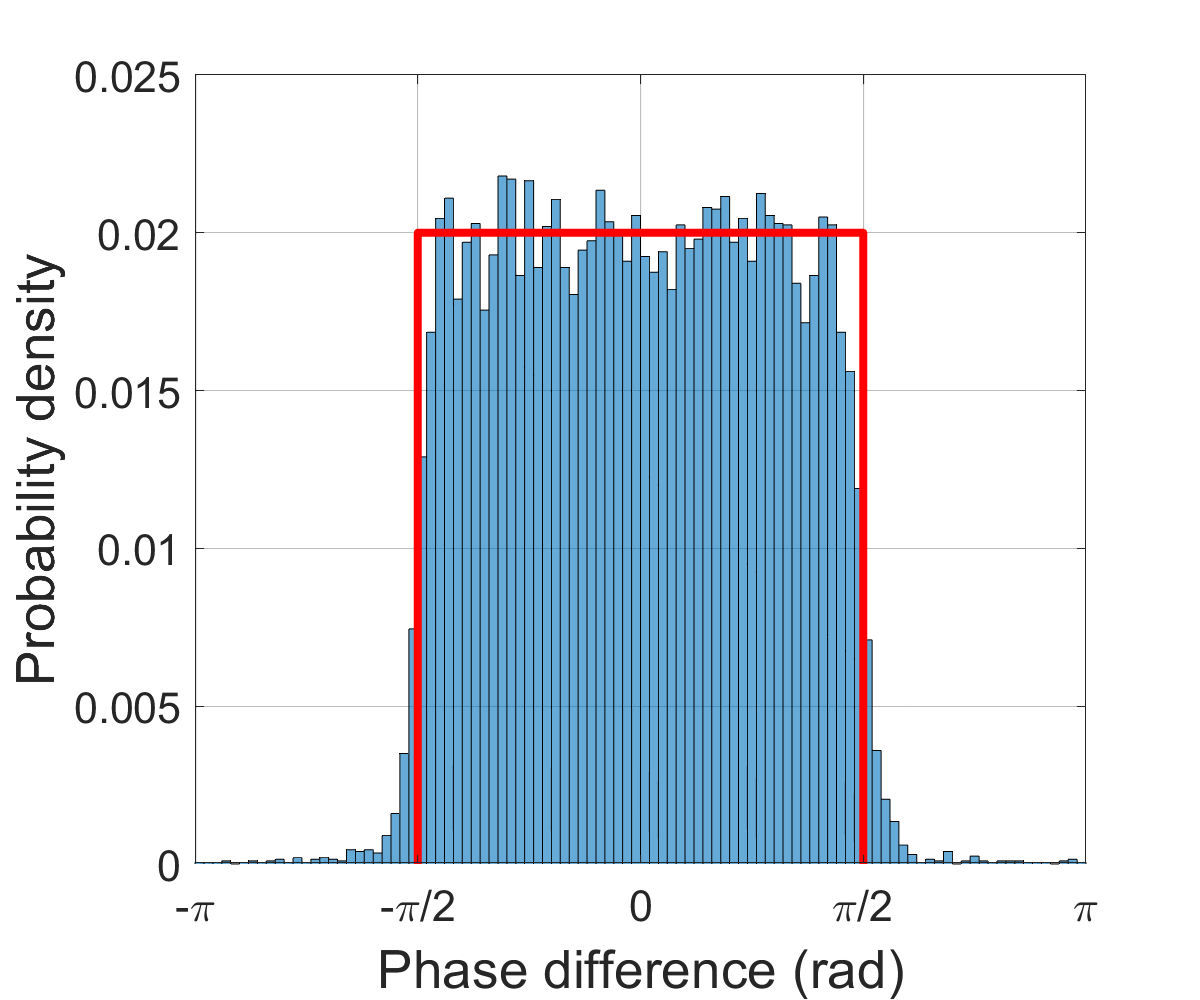}
    }   
    \subfigure[$K = 2,~T = 200000$.]
    {\includegraphics[scale = 0.131]{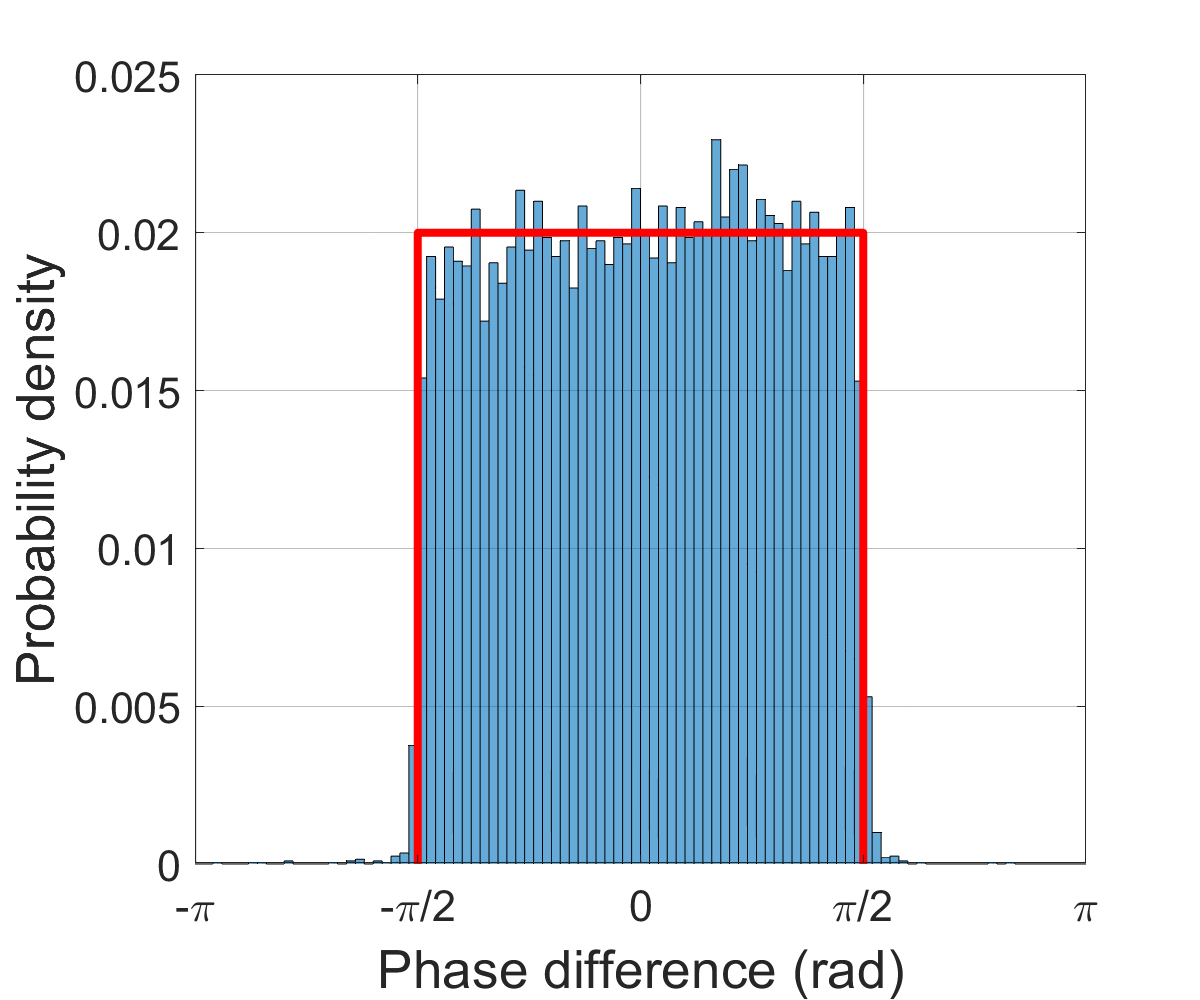}
    }
\quad
    \subfigure[$K = 4,~T = 200$.]
    {\includegraphics[scale = 0.131]{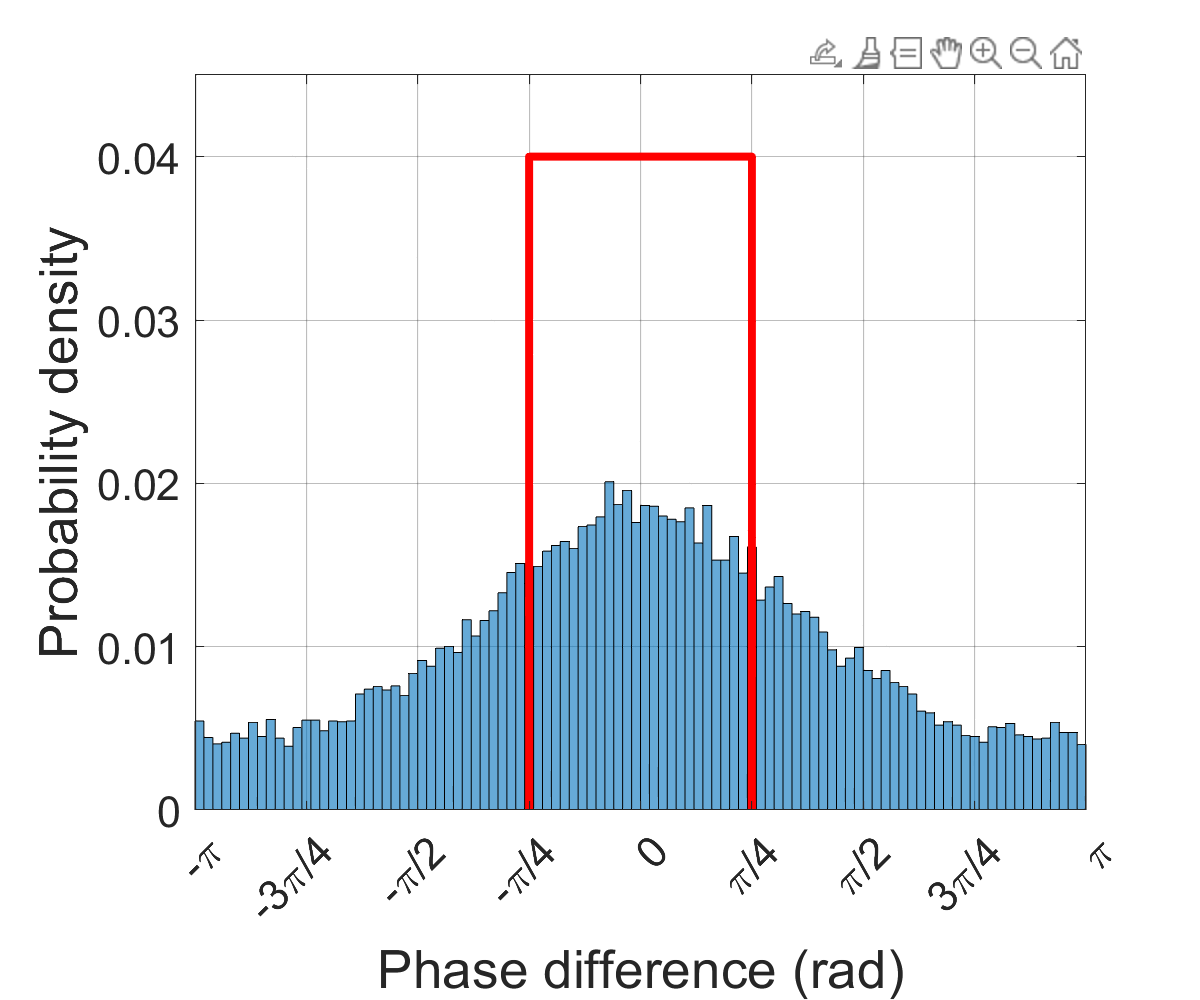}
    }     
     \subfigure[$K = 4,~T = 2000$.]
    {\includegraphics[scale = 0.131]{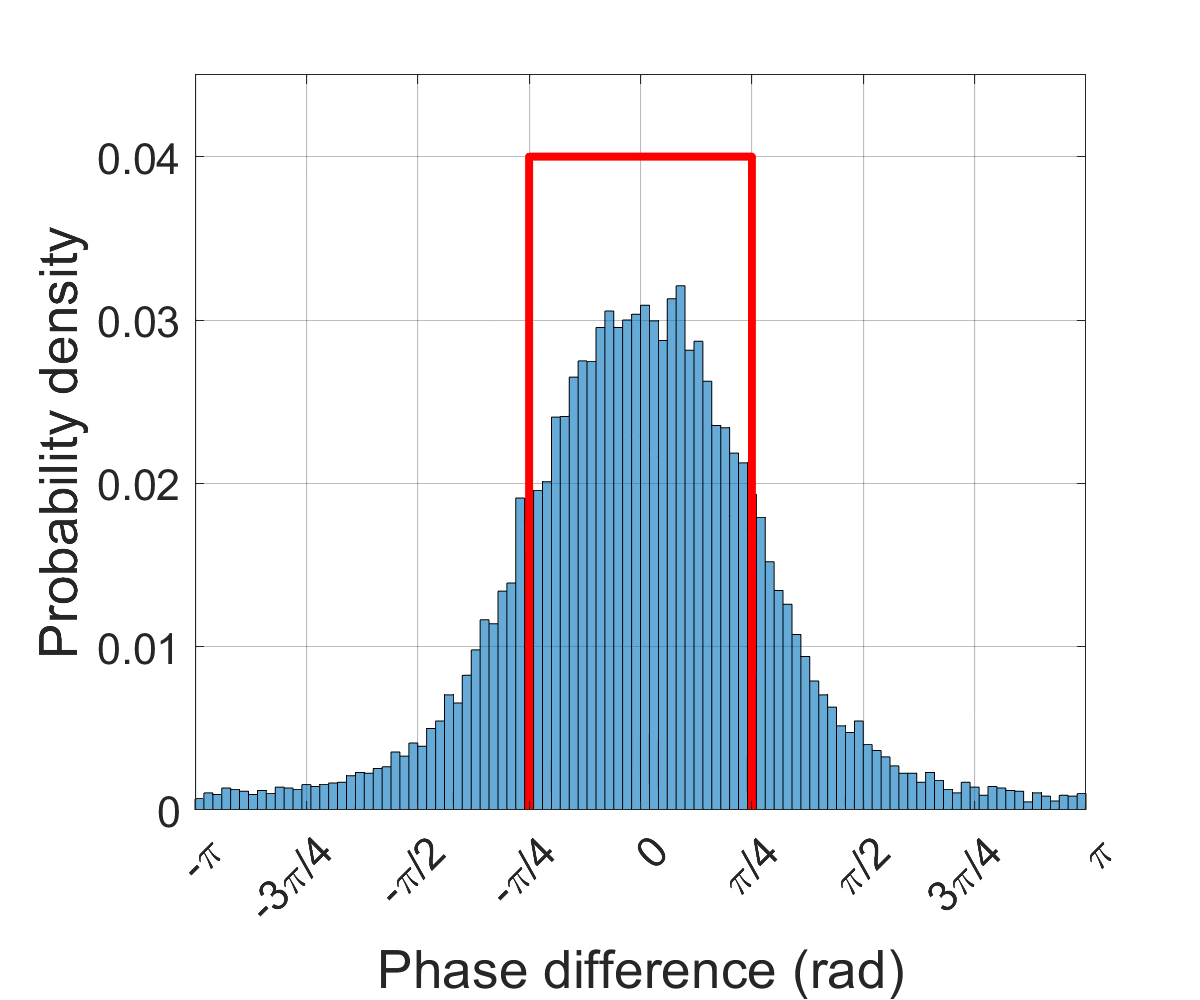}	
    }
\quad
    \subfigure[$K = 4,~T = 20000$.]
    {\includegraphics[scale = 0.131]{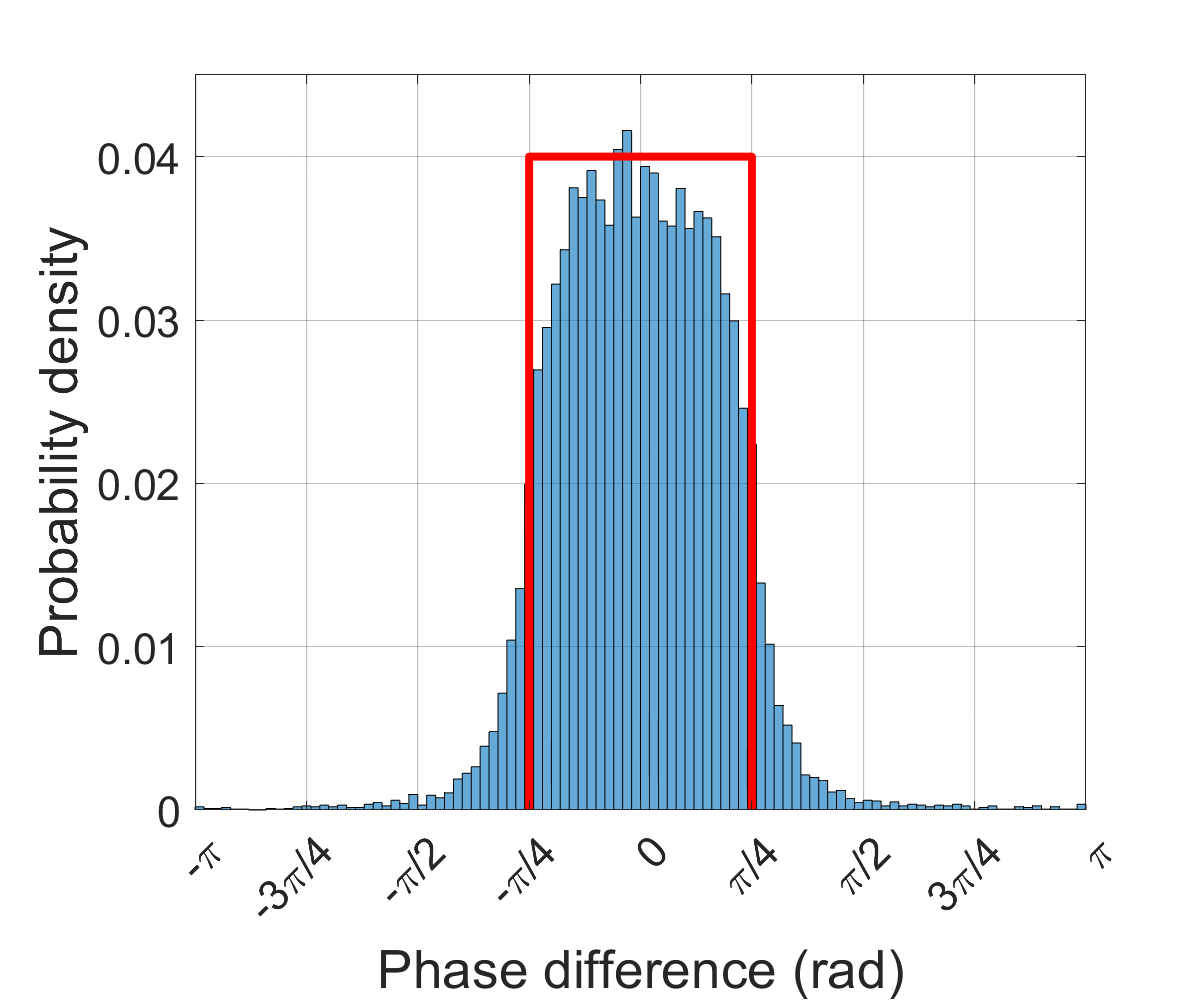}
    }   
    \subfigure[$K = 4,~T = 200000$.]
    {\includegraphics[scale = 0.131]{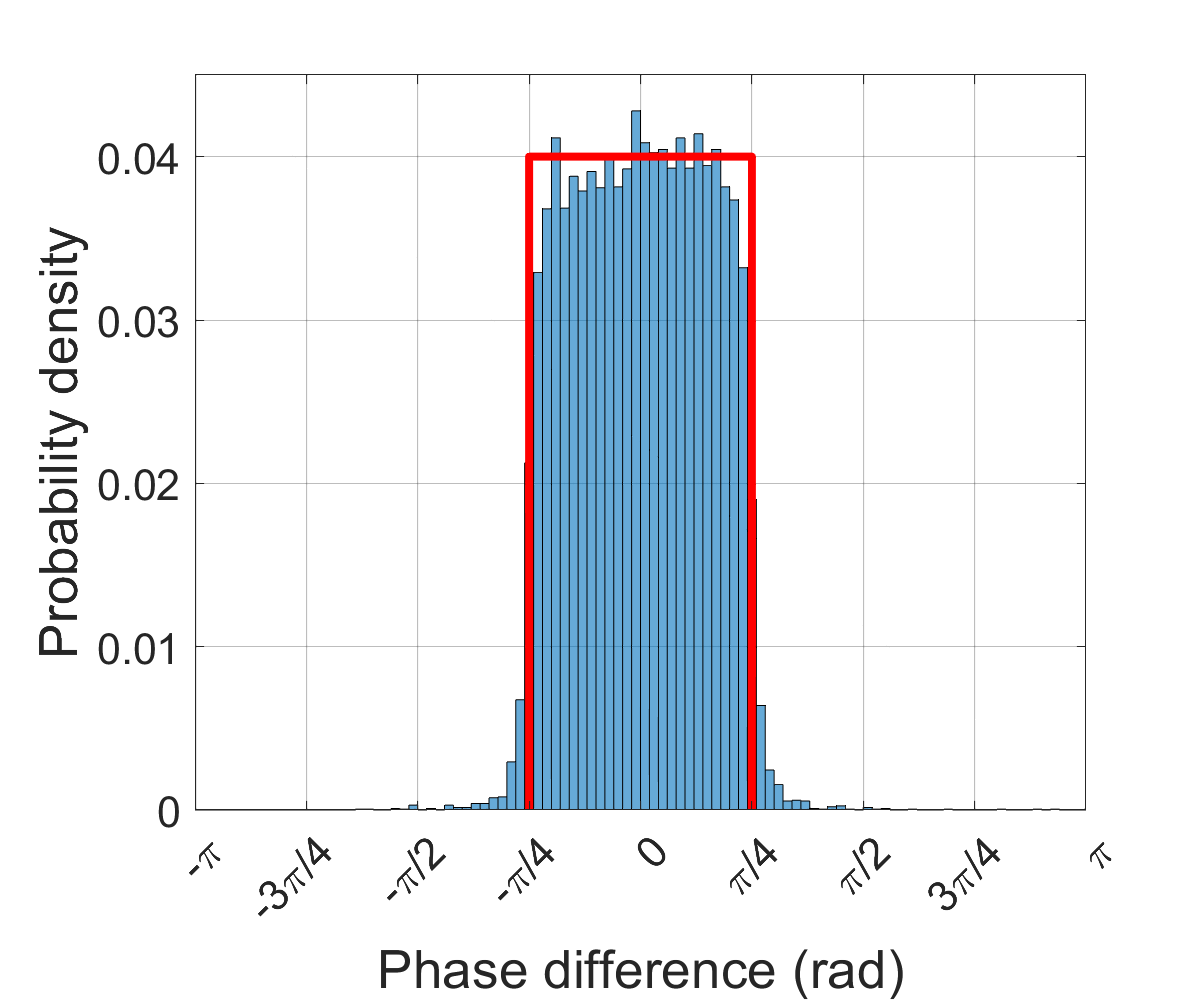}
    }
\caption{Phase alignment performance of the CSM algorithm for $N = 200$, with $K = 2, 4$ and $T = 200, 2000, 20000, 200000$, respectively. The red box indicates the ideal distribution of the phase difference.}
    \label{Performance_of_CSM}
\end{figure*}
We consider significant interference conditions and combine \textbf{Case 2} and \textbf{Case 3} to summarize the expressions of the number of activated IRS elements $SN$ and IPCR that change with $N$:
\begin{equation}
SN =\left\{
\begin{array}{cc}
N, & N< \frac{g_0}{c\rho}+\frac{c^2-1}{2c^2} \\
\frac{g_0}{c\rho}+\frac{c^2-1}{2c^2}, & N \ge \frac{g_0}{c\rho}+\frac{c^2-1}{2c^2}
\end{array}\right..  
\label{SN_result}
\end{equation}

\begin{equation}
\mathrm{IPCR} =\left\{
\begin{array}{cc}
\frac{(g_{0}-Nc\rho)^2+N\rho^2(1-c^2)}{g_{0}^2}, & N< \frac{g_0}{c\rho}+\frac{c^2-1}{2c^2} \\
\frac{1 - c^{2}}{c}\frac{\rho}{g_{0}}\left(1- \frac{1 - c^{2}}{4c}\frac{\rho}{g_0}\right), & N \ge \frac{g_0}{c\rho}+\frac{c^2-1}{2c^2}
\end{array}\right..  
\label{IPCR_result}
\end{equation}

To provide concrete intuition for the scale of IRS elements required in practice, we present several numerical examples based on Equation (\ref{SN_result}). The threshold number of elements $N_{\rm{th}} = \frac{g_0}{c\rho}+\frac{c^2-1}{2c^2}$ determines whether full activation ($S=1$) or partial activation is optimal. We consider typical parameter ranges: the amplitude ratio $g_0/\rho$ between the direct and reflected interference channels, and the phase resolution $K$ that determines $c$ through $c = \sqrt{\pi} \frac{\sin(\omega/2)}{\omega}$ and $\omega = 2\pi / K$. Table \ref{table:N_th_examples} lists $N_{\rm{th}}$ for different settings.
The results indicate that to achieve effective interference suppression (i.e., operate in the saturation regime where $N \ge N_{\rm{th}}$), the required IRS scale ranges from several tens to several hundred elements under common scenarios. 
Furthermore, the theoretical curve of Eq. (\ref{SN_result}) is plotted in Fig. \ref{fig_NS_N} of Section \ref{sec_numerical_simulation} and shows excellent agreement with Monte Carlo simulations, thereby validating the analysis.
\begin{table}[h]
\centering
\caption{Examples of the Threshold Number of IRS Elements $N_{\rm{th}}$}
\label{table:N_th_examples}
\renewcommand{\arraystretch}{1.5} 
\begin{tabular}{c|c|c|c}
\hline
$g_0/\rho$ & $K$ & $c$ & $N_{\rm{th}}$ \\
\hline
10 & 2 & 0.564 & 17 \\
10 & 4 & 0.798 & 12 \\
100 & 2 & 0.564 & 176 \\
100 & 4 & 0.798 & 125 \\
\hline
\end{tabular}
\end{table}



Note that the analytical results in (\ref{SN_result}) and (\ref{IPCR_result}) are derived under the assumption of the optimal proportion value $S$, which will be numerically verified in Section \ref{Verification_of_Theoretical_Analysis}. In fact, however, the proposed PPI algorithm performs a joint search for both the proportion $S$ and the specific combination of activated IRS elements. This additional combinatorial search can yield substantial performance gains that are challenging to characterize analytically. Nevertheless, the overall SINR improvements achieved by the proposed schemes are thoroughly evaluated under various parameter settings in Section \ref{SINR_Improvement}.

\subsection{Legitimate Signal Power Change Ratio (LSPCR)}
The IRS configuration is designed based solely on the interference channel, which is generally independent of the legitimate channel. Therefore, under the optimized IRS configuration, the distributions of the reflected legitimate channels remain unchanged. Therefore, the expected legitimate signal power is given by:   
\begin{equation}
    \begin{split}
\mathbb{E}[|Y-W|^2] = |h_0|^2 + SN\rho^2.  
    \end{split}
\label{expected_legiminate_power}
\end{equation}
The LSPCR is represented as:
\begin{equation}
    \begin{split}
\mathrm{LSPCR} &= \frac{|h_0|^2 + SN\rho^2}{|h_0|^2}= 1+\frac{\rho^2}{|h_0|^2}{SN}.
\label{LSPCR_boost}
\end{split}
\end{equation}

\subsection{Complexity Analysis}
The computational complexity of the proposed scheme arises from two components:
\begin{itemize}
\item CSM algorithm: Requires $T$ RSS measurements with computational complexity $\mathcal{O}(NT)$\cite{CSM_2023}.
\item PPI algorithm: Requires $M/s$ RSS measurements with computational complexity $\mathcal{O}(NM/s)$.
\end{itemize}
Therefore, the total computational complexity of the proposed scheme scales linearly with $N$, which is typically minimal. The total number of RSS measurements is $T+M/s$, typically amounting to a few thousand in practice (see Section~\ref{sec_numerical_simulation}). Since each measurement can be performed with a single symbol duration, the total execution time of the proposed scheme is equivalent to transmitting $T+M/s$ symbols, which typically ranges from milliseconds to seconds.

\section{NUMERICAL RESULTS}
\label{sec_numerical_simulation}
This section presents a comprehensive performance evaluation of the proposed blind anti-interference scheme through Monte-Carlo simulations. The channel model described in Section \ref{sec:system_model} is adopted throughout the simulations. Unless otherwise specified, the system parameters are configured as follows: transmit powers $P_1 = P_2 = 30$ dBm, noise power $\sigma^2 = -94$ dBm. To facilitate a direct and meaningful verification of our theoretical analysis in Eqs. (\ref{SN_result}, \ref{IPCR_result}, \ref{LSPCR_boost}), we fixed the direct channel coefficients to $|h_0|^2 = |g_{0}|^2 = -100$ dB\footnote{Varying $|g_0|$ and $|h_0|$ would cause the metrics to fluctuate between the piecewise regimes defined by Eqs. (\ref{SN_result}, \ref{IPCR_result}, \ref{LSPCR_boost}) for a given $N$, obscuring the comparison between theoretical and simulated curves. This approach of using fixed parameters to isolate and validate a specific theoretical relationship is common in analytical performance evaluation \cite{fixed_parameters, goldsmith2005wireless}.} with random phases. Consequently, in the absence of IRS, the original signal-to-noise ratio (SNR) and signal-to-interference ratio (SIR) are calculated as $\mathrm{SNR} = \frac{P_1|h_0|^2}{\sigma^2} = 24$ dB and $\mathrm{SIR} = \frac{P_1|h_0|^2}{P_2|g_0|^2} = 0$ dB, respectively, indicating that interference constitutes the dominant performance limitation. 
The reflection channel coefficients are modeled as $h_n, g_{n} \overset{\text{i.i.d.}}{\sim} \mathcal{CN}(0,-140~ \text{dB}), n\in \mathcal{N}$, i.e., $\rho^2 = -140~ \text{dB}$. This setting results in an amplitude ratio $|g_0|/\rho = 100$ between the direct and reflected interference channels, consistent with the practical deployment guidance provided in Section \ref{sec_Deployment_Settings}. Based on the theoretical analysis summarized in Table \ref{table:N_th_examples}, the required number of IRS elements is 176 for a phase quantization level of $K = 2$, and 125 for $K = 4$.
The hyperparameters in the CSM and PPI algorithms are set as follows: $T = 10N$, $s = 0.01$ and $M = 10$, incurring a total $10N+1000$ samples. The achievable SINR is computed using equation (\ref{eq:sinr}). 
For comparative analysis, four benchmark schemes are implemented:
\begin{itemize}
   \item \textbf{Upper Bound}: Assuming that interference is completely suppressed, i.e., interference equals zero.
   \item \textbf{Random Phase}: Employs randomly generated IRS phase configurations, representing the baseline scenario without prior information.
   \item \textbf{Random-Min Sampling (RMS)}: From the dataset $\{[\boldsymbol{\theta}_t, |W_t|^2]\},  t = 1, 2, \dots, T$ used in the CSM algorithm, selects the IRS phase configuration that minimizes the measured power, i.e., \[
   \boldsymbol{\theta}^{\text{RMS}} = \arg \min_{\boldsymbol{\theta}_t} |W_t|^2.
   \]
   \item \textbf{CPP with CSI}: Given perfect CSI of the legitimate channels $h_n, n \in \mathcal{N}$, applies continuous phase pursuit and rounds the ideal continuous phases to the nearest discrete values\footnote{Under strong interference conditions, accurate estimation of the legitimate channel requires extensive pilot transmission, introducing significant overhead and delay, which may be practically infeasible.}.
\end{itemize}
All performance results are averaged over 1000 independent channel realizations.

\subsection{Verification of Theoretical Analysis}
\label{Verification_of_Theoretical_Analysis}
This subsection verifies the analytical results derived in Section \ref{Sec_Performance_Analysis} by comparing simulation outcomes with theoretical values.

Fig.~\ref{Performance_of_CSM} illustrates the phase alignment performance of the CSM algorithm. As the sample size $T$ increases, the distribution of phase differences converges toward the ideal distribution. With finite samples (e.g., $T=2000$), the phase alignment remains imperfect, and larger values of $K$ exacerbate this misalignment. Such inaccuracies in phase alignment can degrade the performance of the proposed scheme to some extent, for which the impact of $T$ on the achievable SINR is further examined in Fig.~\ref{fig_SINR_T}.

\begin{figure}[t]
\centering
\includegraphics[width=0.5\textwidth]{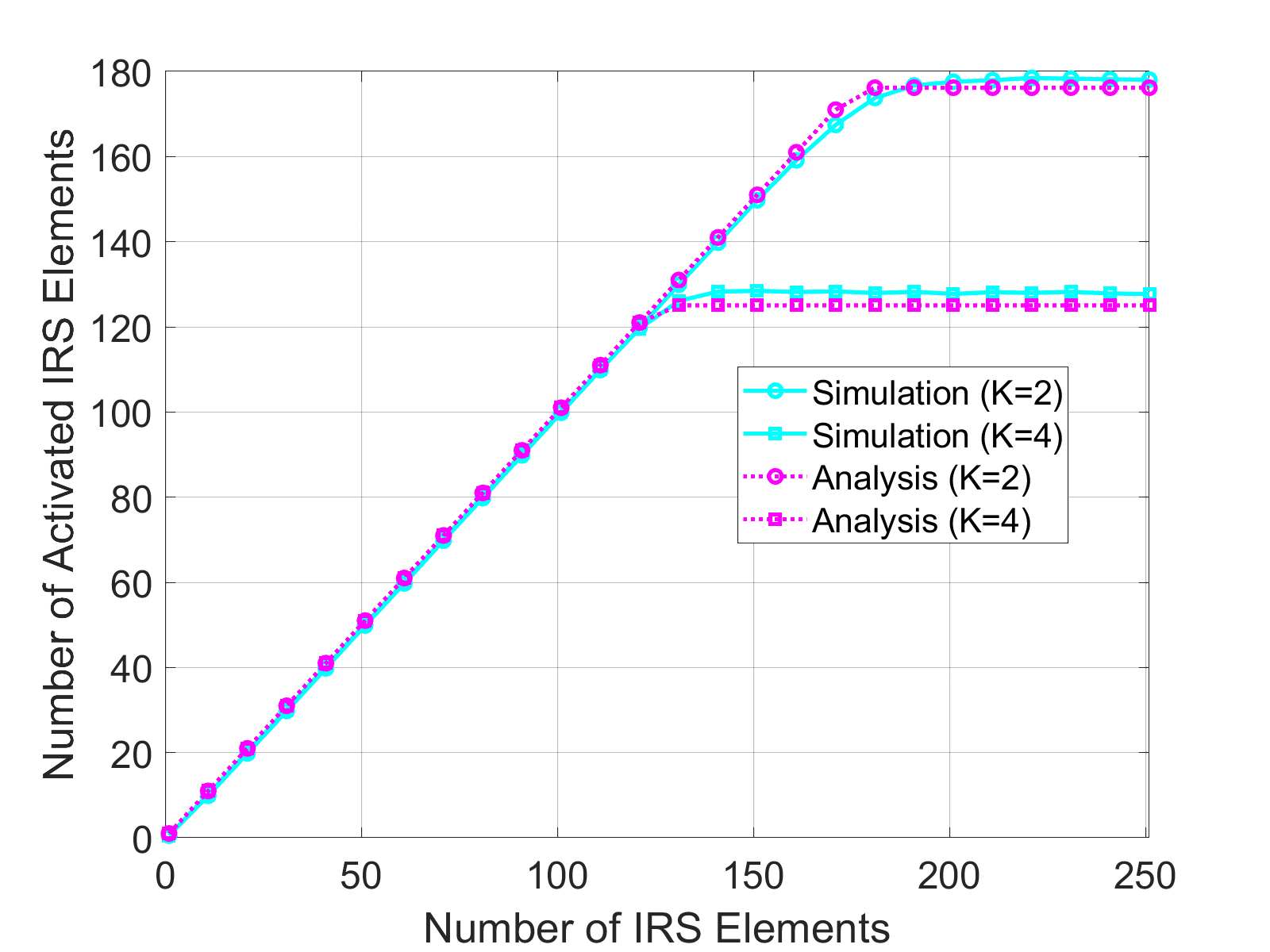}
\caption{Number of activated IRS elements $SN$ versus the total number of IRS elements $N$, for different numbers of phase levels $K = 2, 4$.}
\label{fig_NS_N}
\end{figure}

\begin{figure}[ht]
\centering
\includegraphics[width=0.5\textwidth]{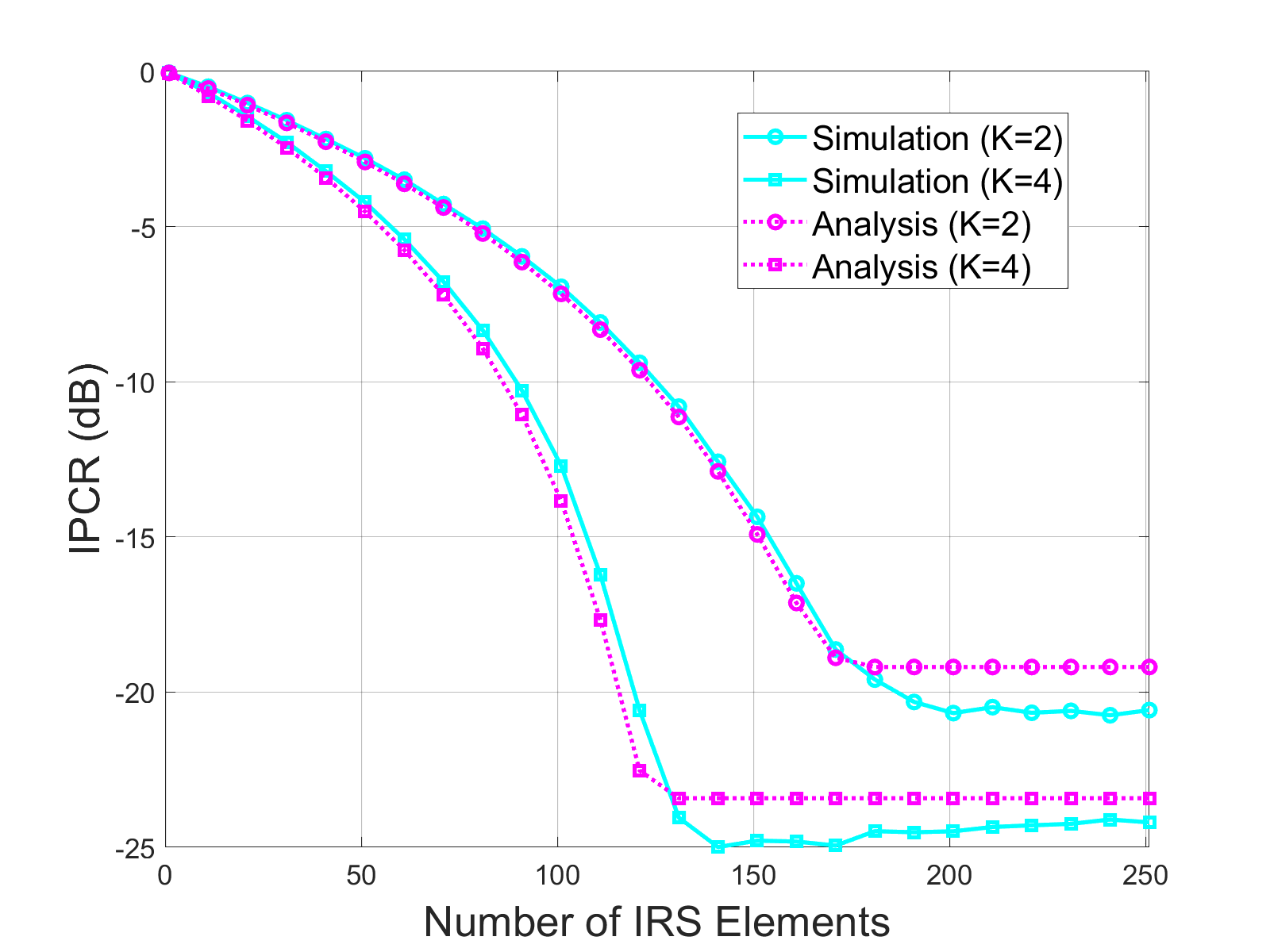}
\caption{$\mathrm{IPCR}$ v.s. $N$, for $K = 2, 4$.}
\label{fig_IPCR_N}
\end{figure}

\begin{figure}[ht]
\centering
\includegraphics[width=0.5\textwidth]{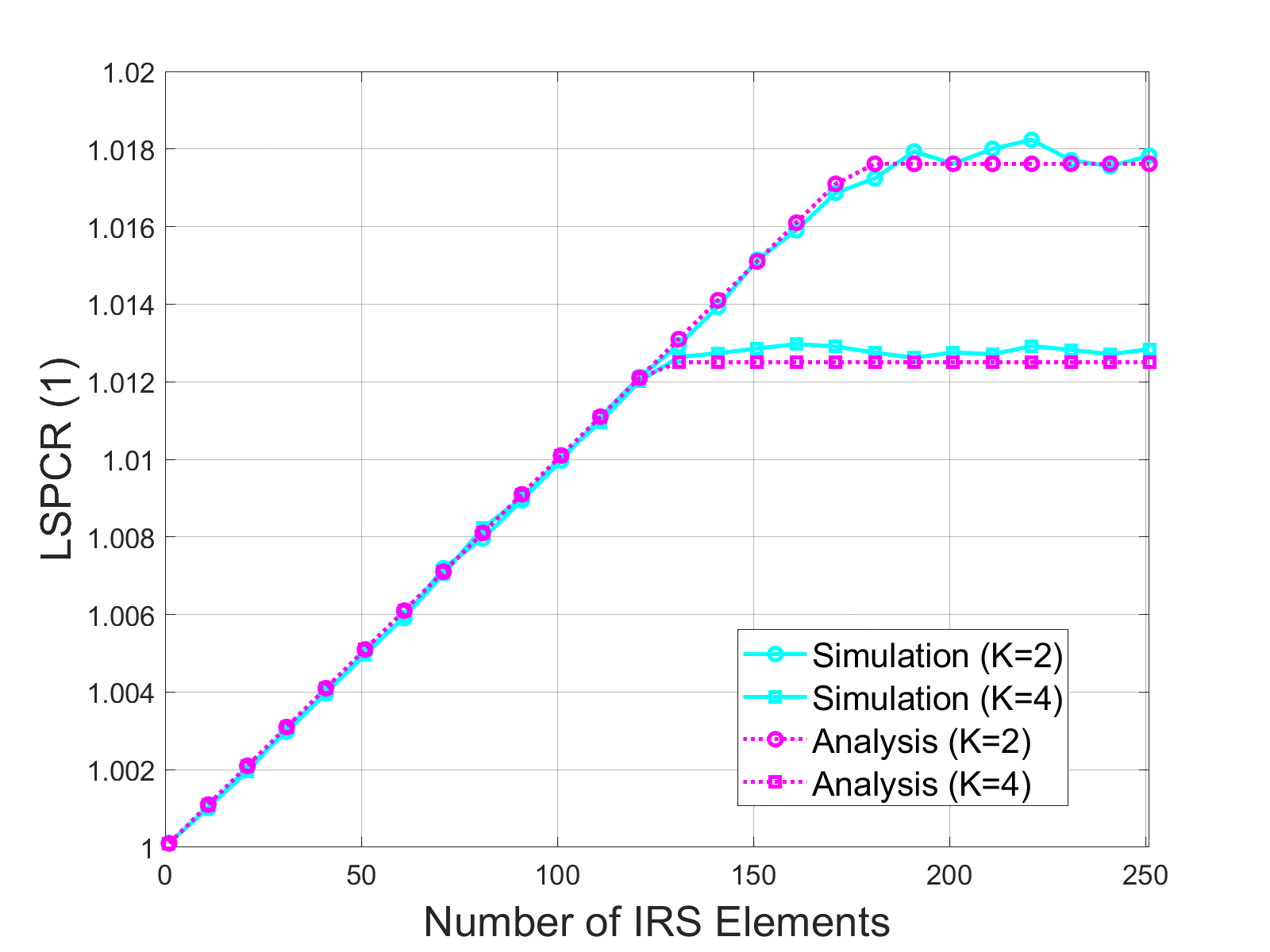}
\caption{$\mathrm{LSPCR}$ v.s. $N$, for $K = 2, 4$.}
\label{fig_LSPCR_N}
\end{figure}

Figs.~\ref{fig_NS_N}, \ref{fig_IPCR_N}, and \ref{fig_LSPCR_N} compare simulation results with theoretical values for $SN$, IPCR, and LSPCR, respectively. The analytical values are computed using Eqs.~(\ref{SN_result}), (\ref{IPCR_result}), and (\ref{LSPCR_boost}). To validate the theoretical analysis via simulation, the third step of the PPI algorithm is modified: instead of selecting the IRS configuration that minimizes the received interference power, we choose the proportion $S$ that achieves the minimum average interference power, and calculate the corresponding average IPCR and LSPCR. Additionally, to guarantee the phase alignment performance of the CSM algorithm, the number of samples is set to $T=100N$.

Overall, the simulation results in Figs.~\ref{fig_NS_N}--\ref{fig_LSPCR_N} align closely with the theoretical predictions, with minor discrepancies attributable to imperfections in the CSM algorithm. In Fig.~\ref{fig_NS_N}, the number of activated IRS elements initially increases with $N$ and eventually saturates. The rising trend indicates insufficient IRS deployment, leading to full activation of all elements. For practical applications, it is essential to deploy a sufficient number of IRS elements (operating in the saturation regime of Fig.~\ref{fig_NS_N}) to achieve effective interference suppression.
Fig.~\ref{fig_IPCR_N} shows that the IPCR can be reduced below $-20$ dB, indicating substantial interference suppression. In addition to the theoretical gain from searching for the optimal $S$, the proposed scheme further improves performance by additionally searching over combinations of activated IRS elements, as will be demonstrated in Section \ref{SINR_Improvement}. Moreover, Fig.~\ref{fig_LSPCR_N} reveals that the legitimate signal power is only weakly enhanced. This is because the configuration of IRS is independent of the legitimate channel.

\begin{figure}[t]
\centering
\includegraphics[width=0.5\textwidth]{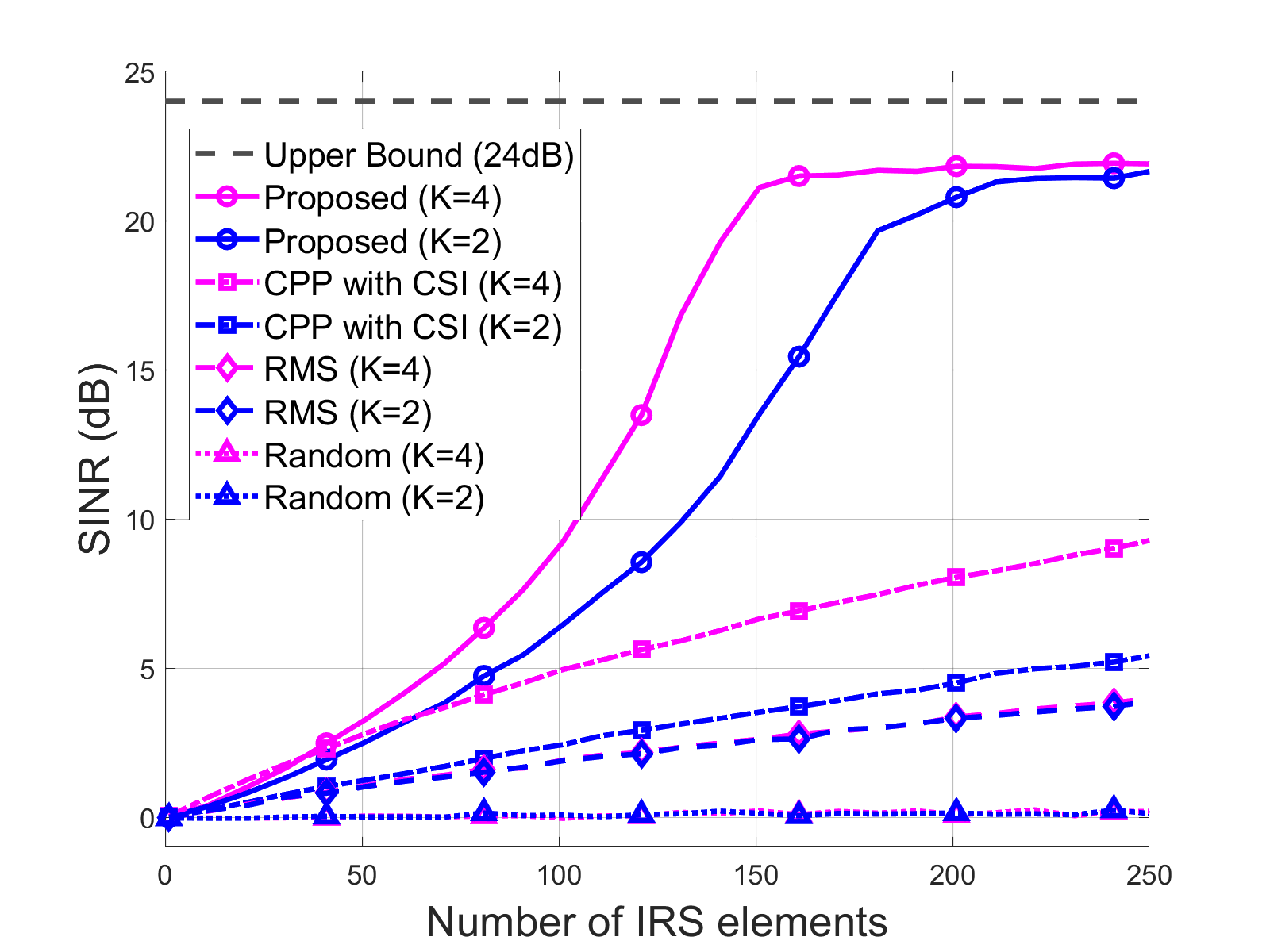}
\caption{$\mathrm{SINR}$ v.s. $N$, for $K = 2, 4$.}
\label{fig_SINR_N}
\end{figure}

\begin{figure}[t]
\centering
\includegraphics[width=0.5\textwidth]{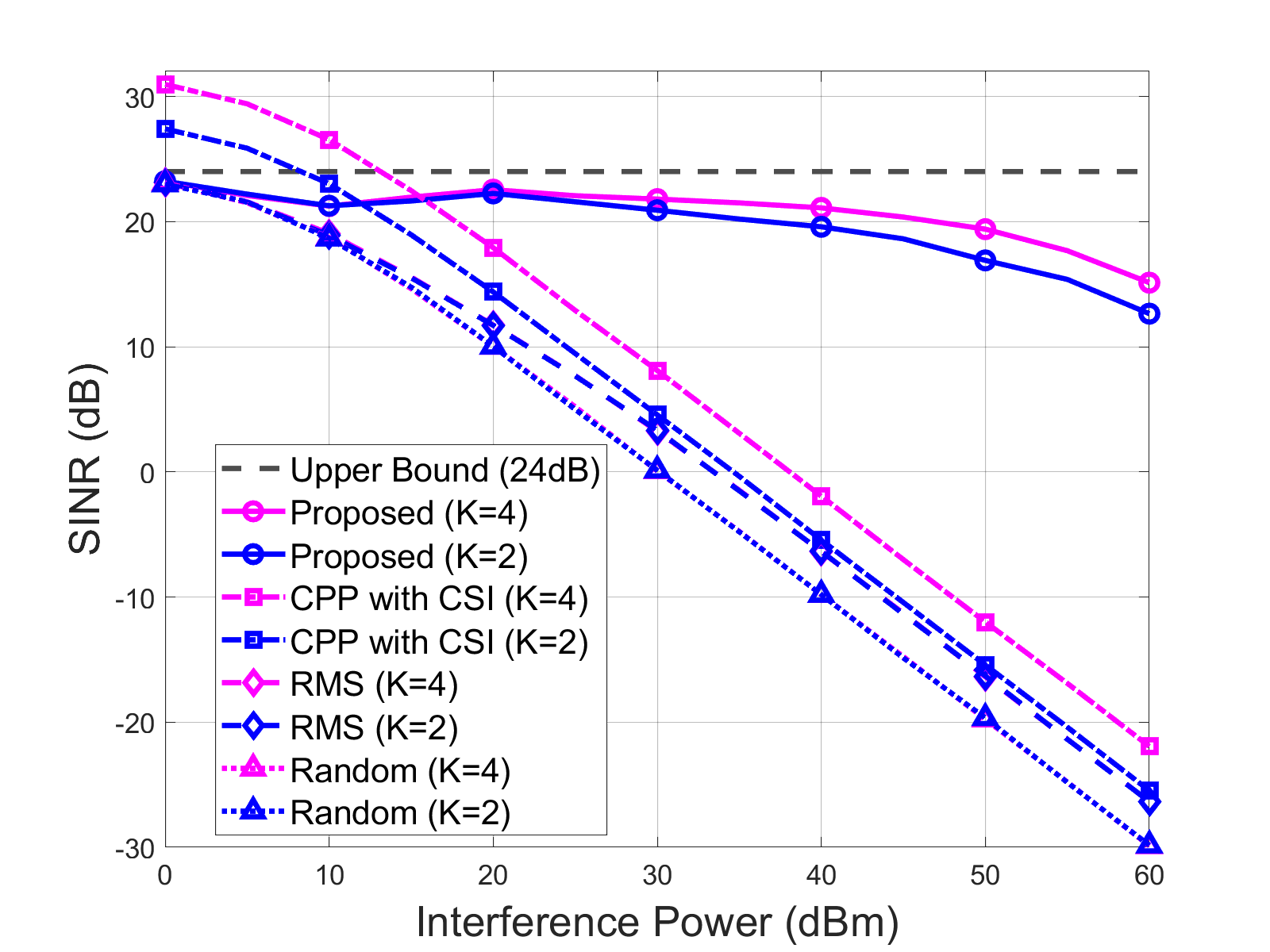}
\caption{$\mathrm{SINR}$ v.s. $P_2$, for $K = 2, 4$.}
\label{fig_SINR_P2}
\end{figure}

\begin{figure}[t]
\centering
\includegraphics[width=0.5\textwidth]{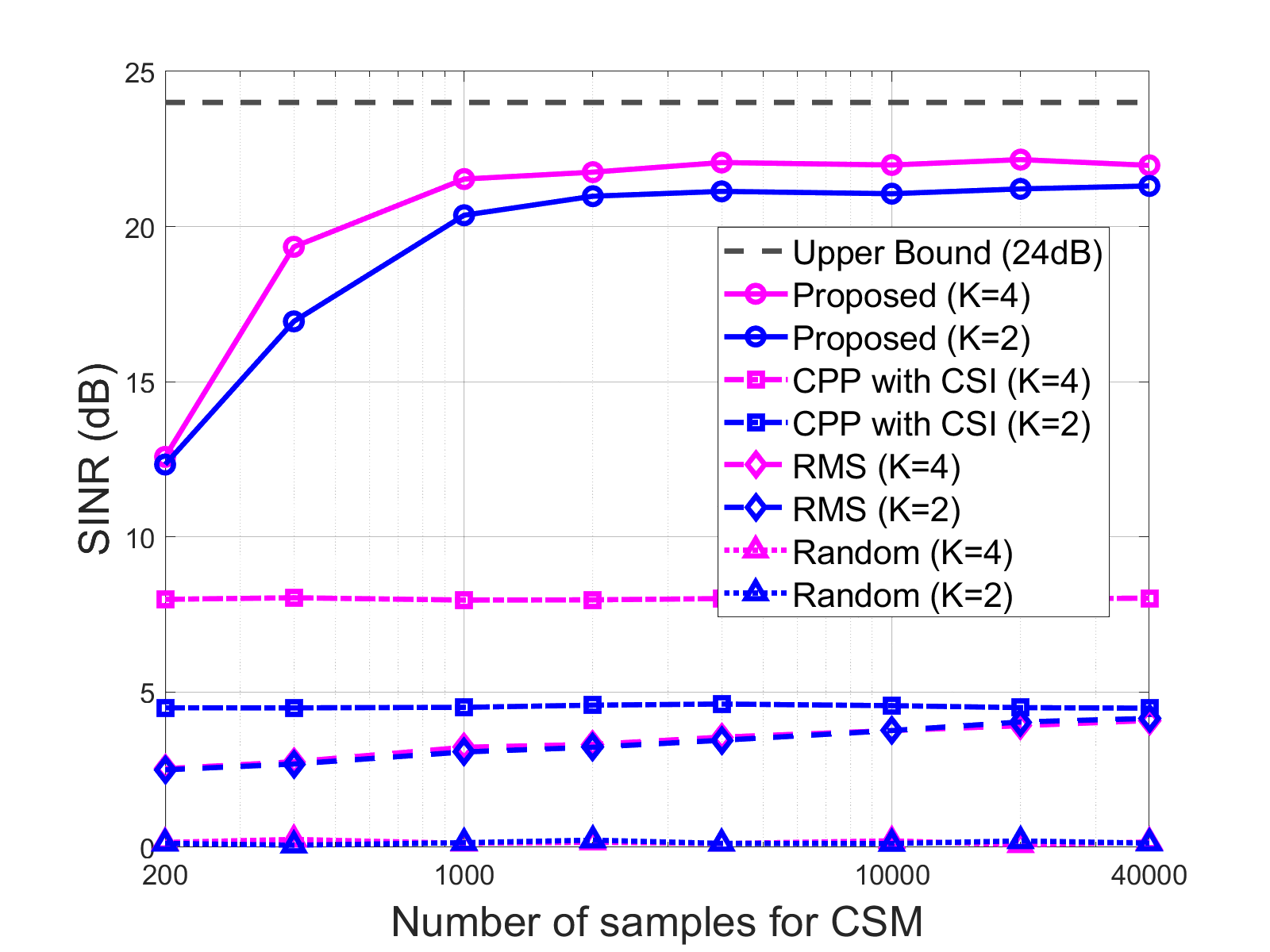}
\caption{$\mathrm{SINR}$ v.s. $T$, for $K = 2, 4$.}
\label{fig_SINR_T}
\end{figure}

\subsection{SINR Improvement}
\label{SINR_Improvement}
This subsection evaluates the achievable SINR of the proposed scheme and several benchmark methods, focusing on the influences of the number of IRS elements $N$, the interference power $P_2$, and the sample size $T$.

Fig. \ref{fig_SINR_N} illustrates the SINR performance as a function of $N$, where the initial SINR (without IRS) is approximately 0 dB. Several key observations can be drawn from the results. As $N$ increases, the proposed scheme achieves rapid SINR improvement before saturation, with gains exceeding 20 dB. In the saturation regime, the gap between the proposed scheme and the upper bound is only about 2 dB, which means that with the proposed scheme, the interference power can be suppressed to the noise level.
In contrast, although the CPP method employs perfect legitimate CSI to maximize the legitimate signal power, it performs suboptimally under strong interference, where interference suppression becomes more critical than signal enhancement. Moreover, acquiring accurate CSI under such strong-interference conditions entails considerable practical difficulties. The proposed scheme, however, attains substantial SINR improvement without requiring any prior CSI. Meanwhile, the RMS method offers only marginal performance gains, which is due to its limited exploitation of statistical information from the received power measurements. The random phase scheme enhances both the interference and legitimate signals equally due to the absence of prior information, leading to negligible SINR improvement and thus approximating the initial SINR. Finally, it is noteworthy that the saturation values of SINR are nearly identical for different phase quantization levels $K$, suggesting that similar interference suppression performance can be achieved with coarser phase quantization by increasing the number of IRS elements. This trade-off offers valuable insights for designing cost-effective IRS-assisted systems in practice.

Fig.~\ref{fig_SINR_P2} shows the SINR as a function of interference power $P_2$, where the random phase curve approximates the initial SINR. As $P_2$ increases, the SINR of all schemes declines, with the proposed method showing the slowest decrease. Under strong interference ($P_2>15$ dBm), the proposed scheme achieves the largest improvement, e.g., raising SINR from $-30$ dB to about $15$ dB (a $45$ dB gain) at $P_2 = 60$ dBm. This demonstrates its robustness in high-interference regime where other methods falter.
As interference diminishes, the curves of the proposed, RMS, and random phase schemes converge, and the CPP method outperforms the proposed scheme, reflecting the reduced need for interference suppression and the increasing need for legitimate signal enhancement. This suggests the value of dynamically switching between interference suppression and signal enhancement modes under varying levels of interference intensity—a direction for future study.

Fig.~\ref{fig_SINR_T} depicts the effect of sample size $T$ on SINR. The proposed scheme’s SINR rises quickly with $T$ before saturation, indicating that a moderate $T$ (e.g., $T=1000$) suffices for near-optimal performance, supporting the scheme’s practical efficiency. By contrast, the RMS scheme improves slowly with $T$. The CPP and random phase methods remain unaffected by $T$, as they do not rely on the interference suppression process.

\subsection{Performance in Multi-UE Scenarios}
\label{SINR_InterUEdist}
As discussed in Remark \ref{remark_multi_UE}, in multi-UE scenarios, a given UE also receives interference signals reflected by the IRSs serving other UEs. To evaluate the performance of the proposed scheme under such conditions, we consider a scenario with two UE-IRS pairs (denoted as UE1–IRS1 and UE2–IRS2) simultaneously executing the interference suppression procedure. Both pairs share identical channel model parameters. 
Owing to symmetry, we analyze only the performance of UE1. Let $d_{11}$ denote the distance between IRS1 and UE1, and $d_{21}$ denote the distance between IRS2 and UE1. Since each UE is deployed near its associated IRS, $d_{21}$ approximates the inter-UE distance. Based on the model in Section~\ref{sec_Deployment_Settings}, the expected power of the interferer–IRS2–UE1 channel, relative to that of the interferer–IRS1–UE1 channel, is attenuated by a factor of $(d_{21}/d_{11}) ^ 2 $. 

Fig.~\ref{fig_SINR_interUEdist} illustrates the impact of $d_{21}$ on the SINR of UE1, with $d_{11}$ fixed at 0.5~m. The results show that the performance of the proposed scheme degrades only when $d_{21}$ is very small (i.e., $d_ {21}<2$ m). In typical deployments where inter-UE distances exceed 2~m, the interference contributed by other IRSs becomes negligible. This confirms the effectiveness and robustness of the proposed scheme in multi-UE settings.

\begin{figure}[t]
\centering
\includegraphics[width=0.5\textwidth]{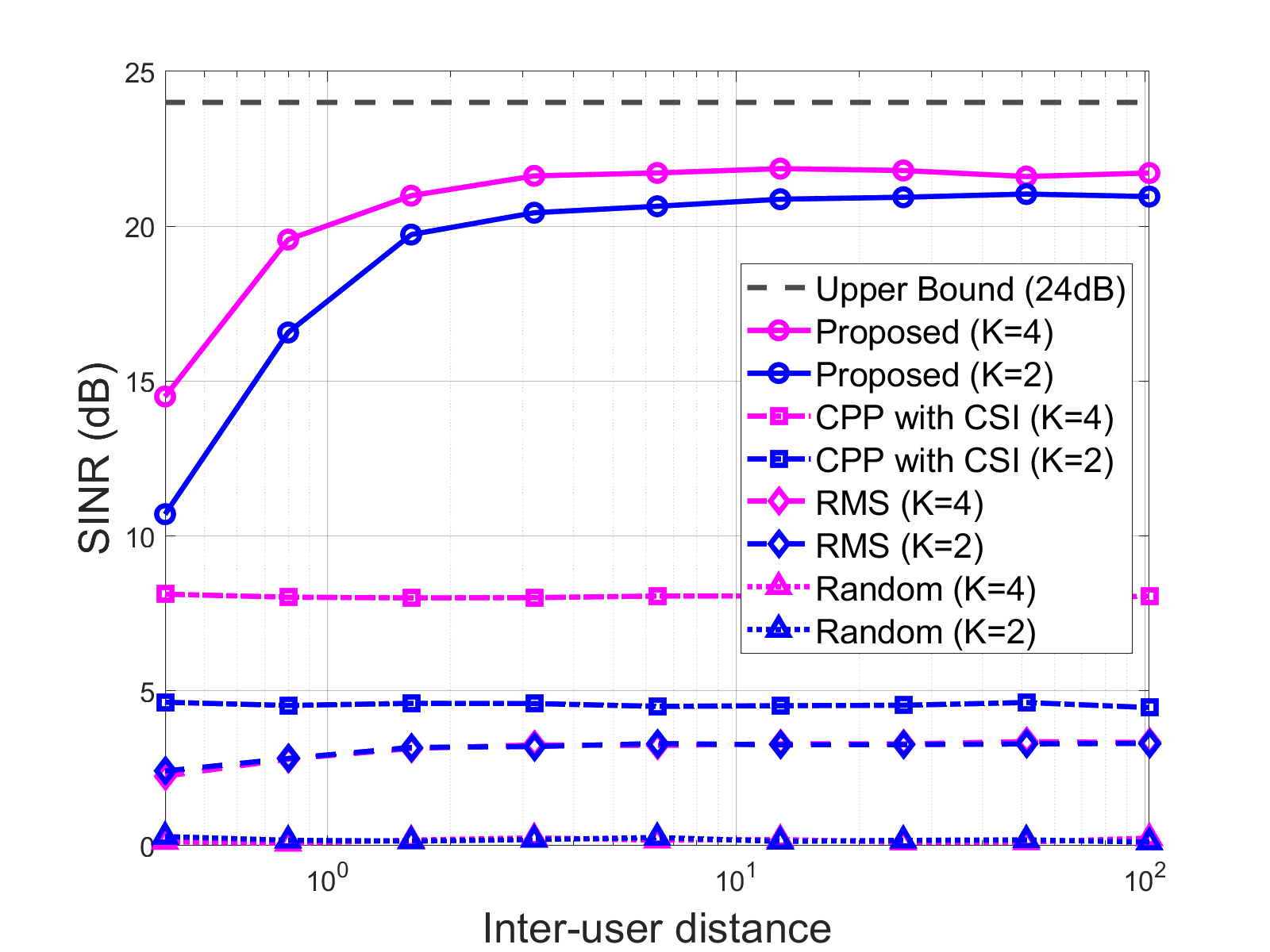}
\caption{$\mathrm{SINR}$ v.s. $d_{21}$.}
\label{fig_SINR_interUEdist}
\end{figure}
\begin{figure}[t]
\centering
\includegraphics[width=0.5\textwidth]{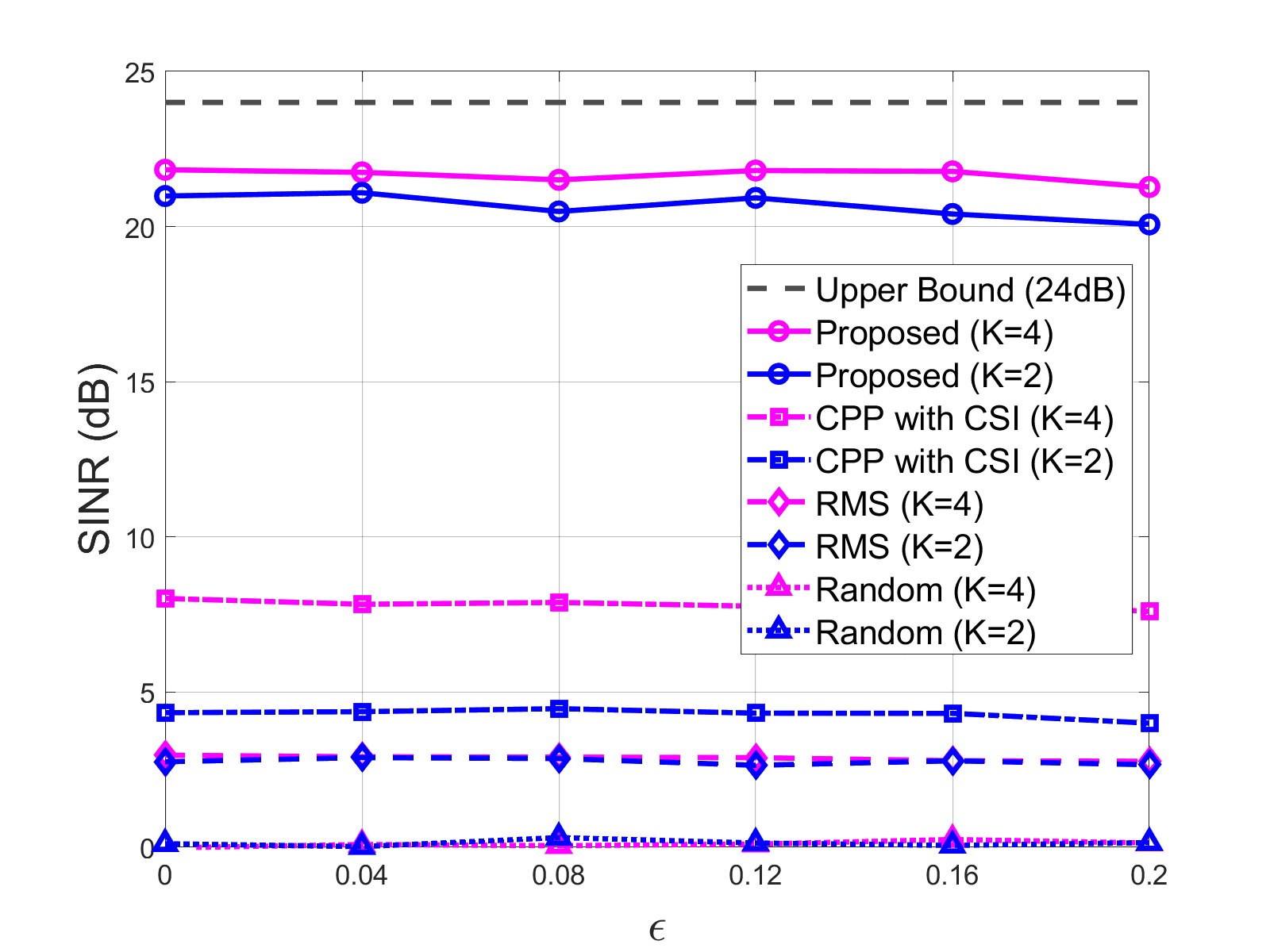}
\caption{$\mathrm{SINR}$ v.s. $\epsilon$.}
\label{fig_SINR_E}
\end{figure}

\subsection{Impact of Imperfect IRS Reflection Coefficients}
\label{SINR_Imperfertion_level}

To evaluate the impact of imperfections in the IRS reflection coefficients—resulting from factors such as insertion loss, quantization errors, coupling effects, and thermal drift—we adopt a comprehensive error model. The insertion loss is set to a fixed value of $\beta = 0.8$. The combined effect of quantization errors, mutual coupling, and thermal drift is modeled as an additive error to the ideal reflection coefficient.
Specifically, the actual reflection coefficient of the $n$-th IRS element when the $k$-th discrete phase shift is applied is given by:
\begin{equation}
\Gamma_{n,k} = \beta e^{j \phi_k} + \epsilon_{n,k},
\end{equation}
where $\phi_k \in \Phi_{K}$ denotes the $k$-th discrete phase shift, and $\epsilon_{n,k} \sim \mathcal{CN}(0, \epsilon)$ represents the additive error, with $\epsilon$ defined as the \emph{imperfection level}. To ensure the magnitude of the reflection coefficient does not exceed unity, any coefficient with $|\Gamma_{n,k}| > 1$ is renormalized as:
\begin{equation}
\Gamma_{n,k} \leftarrow \frac{\Gamma_{n,k}}{|\Gamma_{n,k}|}, \quad \text{if } |\Gamma_{n,k}| > 1.
\end{equation}

Fig.~\ref{fig_SINR_E} illustrates the SINR performance as a function of the imperfection level $\epsilon$, which varies from 0.01 to 0.2. The results show that the performance of the proposed scheme remains stable across different values of $\epsilon$, confirming its robustness against practical imperfections in the IRS reflection coefficients, as discussed in Remark \ref{remark_error}.

\section{Conclusion}
\label{sec_conclusion}
This paper presented a novel blind interference suppression scheme for IRS-assisted wireless communication systems that operates without requiring any CSI. The proposed approach integrates the CSM algorithm for initial phase alignment with a PPI algorithm to effectively suppress interference.
Theoretical performance analysis of the scheme was derived and validated through numerical simulations.
The numerical results under various parameter settings demonstrate that the proposed method reduces the interference power to the noise level, achieving significant improvement in SINR and outperforming existing benchmark schemes under typical interference scenarios.
The proposed scheme offers several key advantages: it does not require prior CSI, supports practical IRS implementations with discrete and imperfect reflection coefficients, scales effectively to multi-UE scenarios, and maintains low computational complexity. These features make it highly suitable for real-world deployment in dynamic interference environments.
Future work may investigate adaptive strategies that dynamically switch between interference suppression and signal enhancement modes under time-varying interference conditions.

\appendix
\subsection*{1. Mean of $W$}
Since the expectation is linear and the variables are independent, we have:
\[
\mathbb{E}[W] = g_{0}
 + \sum_{n=1}^N a_n\mathbb{E}[|g_{n}| e^{j\phi_n}] = g_{0} + SN\cdot\mathbb{E}[|g_{n}| e^{j\phi_n}].
\]
For each $|g_{n}| e^{j\phi_n}$, we have:
\[
\mathbb{E}[|g_{n}| e^{j\phi_n}] = \mathbb{E}[|g_{n}|
] \cdot \mathbb{E}[e^{j\phi_n}].
\]
Given $\phi_n \sim \mathcal{U}[\pi - \omega/2, \pi + \omega/2]$, we have:
\[    \mathbb{E}[e^{j\phi_n}] = \frac{1}{\omega} \int_{\pi - \omega/2}^{\pi + \omega/2} e^{j\phi} d\phi = -\frac{2\sin(\omega/2)}{\omega}.\]
Given $g_{n}\sim\mathcal{CN}(0,\rho^2)$, so $|g_{n}|\sim \text{Rayleigh}(\rho\sqrt{1/2})$. Therefore, we have:
\[\mathbb{E}[|g_{n}|] = \rho\sqrt{1/2}\sqrt{\pi/2} = \frac{\sqrt{\pi}}{2}\rho.\]
\[\mathbb{E}[|g_{n}|^2] = 2(\rho\sqrt{1/2})^2 = \rho^2.\] 
Thus:
\begin{equation}
\begin{split}
    \mathbb{E}[|g_{n}| e^{j\phi_n}] &= \rho\frac{\sqrt{\pi}}{2} \cdot -\frac{2\sin(\omega/2)}{\omega} \\
    &= -\rho\sqrt{\pi}\cdot \frac{\sin(\omega/2)}{\omega}.
\end{split}
\end{equation}
Define the constants:
\[
c = \sqrt{\pi} \frac{\sin(\omega/2)}{\omega}.
\]
The mean of $W$ is:
\begin{equation}
\begin{split}
\mathbb{E}[W] = g_{0}-SN\rho c.
\end{split}
\label{mean_of_W}
\end{equation}

\subsection*{2. Variance of $W$}
Since $g_{0}$ is constant and $|g_{n}| e^{j\phi_n}, n\in\mathcal{N}$ and $V$ are independent:
\begin{equation}
\begin{split}
\mathrm{Var}[W] &= \sum_{n=1}^N a_n\mathrm{Var}[|g_{n}| e^{j\phi_n}]+\sigma^2 \\
&= SN \cdot \mathrm{Var}[|g_{n}| e^{j\phi_n}]+\sigma^2.
\end{split}
\end{equation}
For each $|g_{n}| e^{j\phi_n}$:
\begin{equation}
\begin{split}
\mathrm{Var}[|g_{n}| e^{j\phi_n}] = \mathbb{E}[(|g_{n}| e^{j\phi_n})^2] - \left(\mathbb{E}[|g_{n}| e^{j\phi_n}]\right)
^2. 
\end{split}
\end{equation}
Since $\mathbb{E}[(|g_{n}| e^{j\phi_n}) ^2] = \mathbb{E}[ |g_{n}|
^2] = \rho^2$ and $\mathbb{E}[|g_{n}| e^{j\phi_n}] ^2 =\rho^2c^2.$ Thus:
\[
\mathrm{Var}[|g_{n}| e^{j\phi_n}] =  \rho^2(1 - c^2).
\]
The total variance is:
\begin{equation}
\begin{split}
\mathrm{Var}[W] = SN \rho^2\left(1 - c^2\right)+\sigma^2.
\end{split}
\label{var_of_W}
\end{equation}

\subsection*{3. Expected Power of $W$}
The expected power is:
\begin{equation}
\begin{split}
\mathbb{E}[|W|^2] = \mathrm{Var}[W] + \left|\mathbb{E}[W]\right|^2.
\end{split}
\label{power_of_W}
\end{equation}
Bt substituting (\ref{mean_of_W}) and (\ref{var_of_W}) into (\ref{power_of_W}), we can get:\[
\mathbb{E}[|W|^2] = A S^2 + B S + C.\]
where:
\[
A =  \rho^2c^2 N^2,~B = N [(1 - c^2) \rho^2 - 2 g_{0} c \rho],~C = g_{0}^2 + \sigma^2.
\]

\def\baselinestretch{1}
\bibliographystyle{IEEEbib}
\bibliography{IEEErefs}
\balance 
\end{document}